\documentclass[reqno,11pt]{amsart}

\usepackage{amsfonts}
\usepackage{amssymb}
\usepackage[latin1]{inputenc}
\usepackage{mathtools}

\usepackage{tikz-cd}        

\newtheorem{thm}{Theorem}[section]
\newtheorem{cor}[thm]{Corollary}
\newtheorem{lemma}[thm]{Lemma}

\newtheorem{prop}[thm]{Proposition}
\theoremstyle{definition}
\newtheorem{defn}[thm]{Definition}
\theoremstyle{remark}

\newtheorem{remark}[thm]{Remark}
\newtheorem{example}[thm]{Example}

\newcommand{\PB}{\left\{\cdot\,,\cdot\right\}}
\newcommand{\PBm}{\left\{\cdot\,;\cdot\right\}}
\newcommand{\pB}[1]{\left\{#1,\cdot\right\}}
\newcommand{\pBm}[1]{\left\{#1\,;\cdot\right\}}
\newcommand{\Pb}[1]{\left\{\cdot\,,#1\right\}}
\newcommand{\pb}[1]{\left\{#1\right\}}
\newcommand{\pbm}[2]{\left\{#1\,;#2\right\}}

\newcommand{\lb}[1]{\[#1\]}
\newcommand{\LB}{[\cdot\,,\cdot]}
\newcommand{\lB}[1]{\left[#1,\cdot\right]}

\renewcommand{\(}{\left(}
\renewcommand{\)}{\right)}
\renewcommand{\[}{\left[}
\renewcommand{\]}{\right]}

\newcommand{\set}[1]{\left\{#1\right\}}

\newcommand{\cA}{\mathcal A}
\newcommand{\fA}{\mathfrak A}
\newcommand{\cB}{\mathcal B}

\newcommand{\cH}{\mathcal H}
\newcommand{\fH}{\mathfrak H}
\renewcommand{\H}[2]{{\mathfrak H}^{(#1)}_{#2}}
\newcommand{\cI}{\mathcal I}
\newcommand{\cJ}{\mathcal J}

\newcommand{\fg}{\mathfrak g}

\newcommand{\cO}{\mathcal O}

\newcommand{\cS}{\mathcal S}

\newcommand{\bbC}{\mathbb C}

\newcommand{\bbK}{\mathbb K}
\newcommand{\bbN}{\mathbb N}

\newcommand{\bbT}{\mathbb T}

\newcommand{\bbZ}{\mathbb Z}

\newcommand{\bF}{\mathbf F}
\newcommand{\bG}{\mathbf G}
\newcommand{\bH}{\mathbf H}

\newcommand{\cl}[1]{\overline{#1}}
\newcommand{\pr}{p_\Pi}
\newcommand{\prp}{p_{\Pi'}}
\newcommand{\h}{\nu}
\newcommand{\q}{{\mathbf q}}

\newcommand{\Sym}{\mathop{\rm Sym}\nolimits}

\newcommand{\leqs}{\leqslant}

\newcommand{\tns}{\otimes}
\newcommand{\pp}[2]{\frac{\partial#1}{\partial#2}}
\newcommand{\p}{\partial}

\newcommand{\Span}{\mathop{\rm Span}\nolimits}

\newcommand{\Id}{\mathrm{Id}}

\newif\ifprivate
\privatefalse

 \numberwithin{equation}{section}

\def\???{\ifprivate {\bf {???}} \marginpar{{\Huge {\bf ?}}}\else \fi}
\numberwithin{equation}{section}

\begin{document}  

\nocite{*} 

\title[Poisson algebras from deformations]{Commutative Poisson algebras from 
deformations of noncommutative algebras}

\author[Mikhailov]{Alexander V. Mikhailov} \address{Alexander V. Mikhailov, School of Mathematics, University of Leeds, UK}\email{a.v.mikhailov@leeds.ac.uk}

\author[Vanhaecke]{Pol Vanhaecke} \address{Pol Vanhaecke, Laboratoire de Math\'ematiques et
  Applications, UMR 7348 CNRS, Universit\'e de Poitiers, 11 Boulevard Marie et Pierre Curie, Téléport 2 - BP 30179,
  86 962 Chasseneuil Futuroscope Cedex, France}\email{pol.vanhaecke@math.univ-poitiers.fr}

\thanks{Partially supported by a the LMS (research in pairs) grant \#42218}

\date{\today}


\begin{abstract}
It is well-known that a formal deformation of a commutative algebra $\cA$ leads to a Poisson bracket on $\cA$ and
that the classical limit of a derivation on the deformation leads to a derivation on $\cA$, which is Hamiltonian
with respect to the Poisson bracket. In this paper we present a generalisation of it for formal deformations of an
arbitrary noncommutative algebra $\cA$. The deformation leads in this case to a Poisson algebra structure on
$\Pi(\cA):=Z(\cA)\times(\cA/Z(\cA))$ and to the structure of a $\Pi(\cA)$-Poisson module on $\cA$. The limiting
derivations are then still derivations of $\cA$, but with the Hamiltonian belong to $\Pi(\cA)$, rather than to
$\cA$. We illustrate our construction with several cases of formal deformations, coming from known quantum
algebras, such as the ones associated with the nonabelian Volterra chains, Kontsevich integrable map, the quantum
plane and the quantised Grassmann algebra.

\end{abstract}

\maketitle

\setcounter{tocdepth}{2}



\section{Introduction}

By a well-known procedure, usually referred to as ``taking the classical limit'', quantum systems become classical
systems, equipped with a Hamiltonian stucture (symplectic or Poisson).  Dirac~\cite{Dirac25} observed that
Heisenberg's noncommutative multiplication of operators in a quantum algebra~$\cA_\hbar$ is a deformation of the
commutative multiplication of functions on phase space. As the deformation parameter, the Planck constant $\hbar$,
tends to zero, the algebra $\cA_\hbar$ tends to the commutative algebra~$\cA$ of functions of phase space and the
commutator of operators $\hat a,\hat b\in\cA_\hbar$ converges to the Poisson bracket
\begin{equation*}
  \{a,b\}=\lim\limits_{\hbar\to 0}\frac{i}{ \hbar}[\hat a,\hat b]
\end{equation*}
of the corresponding functions $a,b\in\cA $ on phase space. Heisenberg's equation with a quantum Hamiltonian $\hat
H\in\cA_\hbar$ tends to the Hamiltonian equation
\begin{equation*}
  \frac{d\hat a}{dt}=\frac{i}{\hbar}[\hat H,\hat a]\quad \xrightarrow{\hbar\to 0}\quad\frac{d a}{dt}=\{ H,a\}
\end{equation*}
and defines the Hamiltonian derivation $\p_H:=\{H,\cdot\}\,:\cA\to\cA$.

A novel approach to quantum algebras has recently be introduced by the first author \cite{AvM20}. These algebras,
which will here be simply called \emph{quantum algebras}, admit by definition a basis of ordered monomials (see
Definition \ref{def:quantum_algebra}). They appear naturally in the study of nonabelian systems
\cite{CMW22,buchstaber2023kdv,CMW24} and by taking the classical limit one obtains the underlying
classical Hamiltonian system, under the assumption that the limiting algebra is a \emph{commutative} algebra.

An example of where the method of taking the classical limit fails is given by the quantum algebra \cite{AvM20}
\begin{equation}\label{eq:CA_intro}
  \cA_q:=\frac{\bbC(q)\langle x_i\rangle_{i\in\bbZ}}{\langle x_{i+1}x_i-(-1)^iqx_ix_{i+1},x_ix_j+x_jx_i\rangle_{|i-j|\neq1}}\;,
\end{equation}
where $\bbC(q)\langle x_i\rangle_{i\in\bbZ}$ is the free algebra on $\dots,x_{-1},x_0,x_1,\dots$. It is clear that
for no value of the deformation parameter $q$  this algebra becomes
commutative, or even graded commutative.  On $\cA_q$ there is a well defined differential-difference equation
\begin{align*}
  \p_{2} x_\ell&=\frac1{q^{2 }-1}\lb{\fH_{2 },x_\ell}=x_\ell x_{\ell+1}^2-x_{\ell-1}^2x_\ell+x_\ell^2x_{\ell+1}-x_{\ell-2}x_{\ell-1}x_\ell+x_{\ell}x_{\ell+1}x_{\ell+2}-
  x_{\ell-1}x_\ell^2\;
\end{align*}
which is the first member of an infinite hierarchy of odd-degree Volterra type systems. All equations of the
hierarchy can be presented for $m=1,2,\dots$ in the Heisenberg form \cite{CMW24}
\begin{equation*}
  \p_{2m} x_\ell=\frac1{q^{2m}-1}\lb{\fH_{2m},x_\ell}\;.
\end{equation*}
They have non-trivial limits when $q$ goes to any $2m$-th root of unity. The problem is that the Hamiltonian
structure of these limits cannot be obtained by taking the conventional classical limit. It was one of the motivations to undertake this study. We solved this problem and in Section \ref{par:another} we use this example to illustrate our method.

There were several attempts to develop a Hamiltonian description of differential equations on associative algebras
which are not commutative. In the case of a free associative algebra $\cA$ a Lie (``Poisson'') bracket can be
defined on the quotient space $\cA^\natural=\cA/ [\cA,\cA]$ of the algebra $\cA$ over the linear space $[\cA,\cA]$
spanned by all commuators in $\cA$, and elements of $\cA^\natural$ define Hamiltonian derivations of $\cA$
\cite{MS2000CMP}. The linear space $\cA^\natural$ does not admit a structure of a Poisson algebra since a
compatible multiplication is missing.  Farkas and Letzter demonstrated that for any prime Poisson algebra $\cA$,
which is {\em not commutative}, the Poisson bracket must be the commutator in $\cA$, up to an appropriate scalar
factor \cite{FL}.  Consequently, it became widely acknowledged that the definition of a Poisson algebra in the
noncommutative case is too restrictive. Several modifications of the definition, inspired by noncommutative Poisson
and differential geometries, as well as Hamiltonian description of noncommutative differential equations were
explored in \cite{Vandenbergh, Kac, CEG}.

In this paper we propose another approach, where the structure that we put on $\cA$ is that of a Poisson module
over a commutative Poisson algebra $\Pi(\cA)$ which we also construct from the deformation. Specifically, suppose
that $(\cA[[\h]],\star)$ is a (formal) deformation of an associative algebra~$\cA$, whose center is denoted
$Z(\cA)$. The (noncommutative) Poisson algebra $(\cA[[\h]],\LB_\star)$ admits $\cH_\h :=Z(\cA)+\h\cA[[\h]]$ as a
Poisson subalgebra, and on it we can define a rescaled Lie bracket $\LB_\h:=\frac1\h\LB_\star$, which makes
$(\cH_h,\LB_\h)$ into a Poisson algebra. The latter admits $\h\cH_\h$ as a Poisson ideal, so that $\cH_\h/\h\cH_\h$
inherits the structure of a Poisson algebra, i.e., a multiplication and a Poisson bracket. Since
$\cH_\h/\h\cH_\h\simeq Z(\cA)\times(\cA/Z(\cA))$, in a natural way, $\Pi(\cA):=Z(\cA)\times(\cA/Z(\cA))$ is a
Poisson algebra. It is in fact a \emph{commutative} Poisson algebra. The commutative multiplication and the Poisson
bracket on $\Pi(\cA)$ will be denoted by $\cdot$ and $\PB$. Like any Poisson bracket on a commutative algebra, the
bracket is completely specified on generators and can be computed for arbitrary elements by using derivatives (see
\eqref{eq:PQ_bracket}), which is one of the main virtues of Poisson brackets on 
commutative algebras. 

In order to construct the $\Pi(\cA)$-Poisson module structure on $\cA$, we first consider the (noncommutative)
Poisson algebra $(\cH_\h/\h^2\cA[[\h]],\LB_\h)$ as a Poisson module over itself, and then show that
$\h\cA[[\h]]/\h^2\cA[[\h]]\simeq\cA$ is a Poisson submodule. Now $\Pi(\cA)$ can be identified with the quotient of
$\cH_\h/\h^2\cA[[\h]]$ with respect to the Poisson ideal $\h Z(\cA)$ and by reduction $\cA$ becomes a Poisson
module over $\Pi(\cA)$, with actions denoted by $\cdot$ and $\left\{\cdot\,;\cdot\right\}$. Moreover we show that
for any $\bH\in\Pi(\cA)$, $\pBm\bH$ is a derivation of $\cA$, making it into a \emph{Hamiltonian} derivation of $\cA$.
Since $\Pi(\cA)$ is commutative, the Poisson module structure can again easily be computed from the formulas, given
in terms of generators for $\Pi(\cA)$ and of $\cA$.

The construction of $\Pi(\cA)$ and of the Poisson module structure on $\cA$ will be illustrated at length in
several examples (Section \ref{sec:ex1}), associated with quantum algebras; its applications to nonabelian
 systems, related to these quantum algebras will be illustrated in Section \ref{sec:ex2}. We will in this
introduction only describe shortly one example, which is worked out in detail in Sections \ref{par:quantum_plane}
and~\ref{par:quantum_torus_2}.

For concreteness, we illustrate our approach in a simple example, related to the integrable Kontsevich equation
\cite{Wolf2012}. Recall that the (two-dimensional) quantum torus is defined as the quantum algebra
\begin{equation}\label{eq:quantum_torus_intro}
  \cA_q:=\bbT_q[x,y]=\frac{\bbC(q)\langle x,y,x^{-1},y^{-1}\rangle}{\langle yx-qxy\rangle}\;.
\end{equation}
It is a localization of the quantum plane $\frac{\bbC(q)\langle x,y\rangle}{\langle yx-qxy\rangle}$.  On $\cA_q$
there is a hierarchy of commuting derivations, given for $n=1,2,\dots$ by
\begin{equation}\label{Khier}
  \p_nx=\frac{1}{1-q^n}[\fH^{(n)},x]\;,\quad   \p_ny=\frac{1}{1-q^n}[\fH^{(n)},y]\;.
\end{equation}
The quantum Hamiltonians $\fH^{(n)}$ are given for $n=1,2$ by
\begin{align*}
  \fH^{(1)}&=qx^{-1} y^{-1}+q y^{-1}+y+qx+x^{-1}\;,\\
  \fH^{(2)}&=\(\fH^{(1)}\)^2-(1+q)\fH^{(1)}-4q\;.
\end{align*}
The coefficients in the right-hand side of equations (\ref{Khier}) are Laurent polynomials in the
variable~$q$. Consequently, the derivation $\p_n$ possesses a well-defined limit as $q\to\xi$, where $\xi$
represents a primitive $n$-th root of unity. By setting $q=\xi+\h$, we can regard $\cA_{q}$ as an algebra of power
series $\cA[[\h]]$, where $\cA=\cA_\xi$. In the algebra $\cA_\xi$ we have $yx=\xi xy$, implying that the center of
$\cA$ is generated by $x^n$ and $y^n$. Moreover, $X:=\(x^n,\cl0\)$, $Y:=\(y^n,\cl0\)$ and
$W_{i,j}:=\(0,\cl{x^iy^j}\)$, where $0\leqslant i,j< n,\ i+j\neq0$ generate~$\Pi(\cA)$, while the
$\Pi(\cA)$--module $\cA$ is generated by $1,\, x$ and $y$. Table~\ref{tab:intro_1} describes in this case the
Poisson structure $\PB$ on $\Pi(\cA)$ in terms of these generators. The~$\cdot$ action and Lie action $\PBm$ of
$\Pi(\cA)$ on $\cA$ are given by Table~\ref{tab:intro_2}.  These tables can be used to describe the Hamiltonian
structure of the limiting derivation of $\p_n$ on $\cA$.

\bigskip

\begin{table}[h!]
  \def\arraystretch{1.8}
  \setlength\tabcolsep{0.4cm}
  \centering
\begin{tabular}{c|ccc}
   $\PB$&$X$&$Y$&$W_{k,\ell}$\\
  \hline
  $X$&$0$&$-\xi^{-1}n^2XY$&$-\xi^{-1}n\ell  XW_{k,\ell}$\\
  $Y$&$\xi^{-1}n^2XY$&$0$&$\xi^{-1} nkYW_{k,\ell}$\\
  $W_{i,j}$&$\xi^{-1} njXW_{i,j}$&$-\xi^{-1}niYW_{i,j}$&$\(\xi^{jk}-\xi^{i\ell}\)W_{i+k,j+\ell}$\\
\end{tabular}
\bigskip
\caption{}\label{tab:intro_1}
\end{table}
\begin{table}[h!]\label{t2i}
  \def\arraystretch{1.8}
  \setlength\tabcolsep{0.4cm}
\begin{tabular}{c|cc}
   $\cdot$&$x$&$y$\\
  \hline
  $X$&$x^{n+1}$&$x^ny$\\
  $Y$&$xy^n$&$y^3$\\
  $W_{i,j}$&$0$&$0$\\
\end{tabular}
\qquad\qquad
\begin{tabular}{c|cc}
   $\PBm$&$x$&$y$\\
  \hline
  $X$&$0$&$-\xi^{-1}nx^ny$\\
  $Y$&$\xi^{-1}nxy^n$&$0$\\
  $W_{i,j}$&$\(\xi^j-1\)x^{i+1}y^j$&$\(1-\xi^i\)x^iy^{j+1}$\\
\end{tabular}
\bigskip
\caption{}\label{tab:intro_2}
\end{table}
\goodbreak
\noindent Let us consider the case $n=2,\ q(\h)=-1+\h$ and $\cA=\bbC\langle x,y,x^{-1},y^{-1}\rangle/\langle
xy+yx\rangle$. Then

\begin{equation*}
  \frac{\fH^{(2)}}{q(\h)^2-1}=\frac1{\h}\(H^{(2)}_0+\h H^{(2)}_1\)\pmod{\cH_\h}\;,
\end{equation*}
where
\begin{align*}
  H^{(2)}_0&=-\frac1{2}\(x^2+y^2+x^{-2}+y^{-2}-x^{-2}y^{-2}\)\;,\\
  H^{(2)}_1&=-\frac12\(x-y-xy-xy^{-1}-x^{-1}y+x^{-1}y^{-2}-x^{-2}y^{-1}\)\;.
\end{align*}
Let $\bH^{(2)}:=(H^{(2)}_0,H^{(2)}_1)$. Then, writing $U,V$ and $W$ respectively as a shorthand for
$W_{1,0},W_{0,1}$ and $W_{1,1}$,
\begin{equation*}
  \bH^{(2)}=-\frac12\(X+Y+X^{-1}+Y^{-1}-X^{-1}Y^{-1}+U-V-W-Y^{-1}W+X^{-1}Y^{-1}U-X^{-1}Y^{-1}V\)\;.
\end{equation*}
Using Table \ref{tab:intro_2}, we obtain the Hamiltonian equation on the algebra 
$\bbC\langle x,y,x^{-1},y^{-1}\rangle/\langle xy+yx\rangle$
\begin{align*}
  \p_{\bH^{(2)}}x&=
  \pbm{\bH^{(2)}}x=\pp{\bH^{(2)}}Y\cdot(-2xy^2)+\pp{\bH^{(2)}}V\cdot(-2xy^2)+\pp{\bH^{(2)}}W\cdot(-2x^2y)\\
  &=(1-Y^{-2}+X^{-1}Y^{-2})\cdot xy^2-(1+X^{-1}Y^{-1})\cdot xy-(1+Y^{-1}+X^{-1})\cdot x^2y\\
    &=xy^2-xy^{-2}+x^{-1}y^{-2}-xy-x^{-1}y^{-1}-x^2y-x^2y^{-1}-y\;,
\end{align*}
and similarly for $\p_{\bH^{(2)}}y$ (see the detailled calculations in Section \ref{par:quantum_torus_2}).

\goodbreak

The structure of the paper is as follows. We show in Section \ref{sec:poisson} how deformations of an associative
algebra $\cA$ lead to a Poisson algebra $\Pi(\cA)$ and to a $\Pi(\cA)$-Poisson module structure on $\cA$. We show
that the construction is functorial and show its relevance constructing Hamiltonian derivation on~$\cA$ from
Heisenberg derivations on $\cA[[\h]]$. We show in Section \ref{sec:ideals} how a quantised algebra, depending on
one or several parameters, can be viewed naturally as a formal deformation of some associative algebra. Examples of
quantised algebras, the corresponding deformations, Poisson algebras and Poisson modules will be given in Section
\ref{sec:ex1}, together with an application. Examples of nonabelian systems related to these quantised algebras
will be given in \ref{sec:ex2}; we illustrate how the Hamiltonian structure of their limit derivations is obtained
by our methods.

Lastly, in the paper, we do not explicitly consider the complex or Hermitian structure of the Hamiltonians and the
equations. While it can be straightforwardly incorporated in each case, doing so might compromise the clarity and
readability of the expressions.

\section{Poisson algebras and Poisson modules from deformations}\label{sec:poisson}

In this section we show how any deformation of a (not necessarily commutative) associative algebra $\cA$ leads in a
natural way to  a commutative Poisson algebra $\Pi(\cA)$ and  a $\Pi(\cA)$-Poisson module structure on $\cA$.

\subsection{Poisson algebras and deformations}

We first recall the definition of a (not necessarily commutative) Poisson algebra (see for example \cite{FL}). Let
$\cA$ be any (unitary) associative algebra over a commutative ring $R$. For $a,b\in\cA$ their product $a\cdot b$ in
$\cA$ will simply be denoted by $ab$ and their commutator $ab-ba$ by $\lb{a,b}$.

\begin{defn}
A skew-symmetric $R$-bilinear map $\PB:\cA\times\cA\to\cA$ is said to be a \emph{Poisson bracket} on $\cA$ when it
satisfies the Jacobi and Leibniz identities: for all $a,b,c\in\cA$,
\begin{align*}
    (1)&\pb{\pb{a,b},c}+\pb{\pb{b,c},a}+\pb{\pb{c,a},b}=0\;, \quad\hbox{(Jacobi identity)}, \\
    (2)&\pb{a,bc}=\pb{a,b}c+b\pb{a,c}\;, \qquad\qquad\qquad\quad\hbox{(Leibniz identity)}.
\end{align*}
$(\cA,\PB)$ or $(\cA,\cdot,\PB)$ is then said to be a \emph{Poisson algebra (over~$R$)}.  When $\cA$ is commutative
one says that the Poisson algebra $(\cA,\PB)$ is \emph{commutative}.

\end{defn}

\begin{example}\label{exa:commutator}
  Any associative algebra $\cA$ has a natural Poisson bracket, given by the commutator $\pb{a,b}:=[a,b]$. Indeed,
  it is well-known that $\LB$ is a Lie bracket on $\cA$ and one easily checks that the Leibniz identity also
  holds, so $(\cA,\LB)$ is a Poisson algebra. This Poisson bracket is trivial if and only if $\cA$ is commutative.
\end{example}

\begin{example}
  When $(\cA,\PB)$ is a Poisson algebra and $\cB$ is a subalgebra of $\cA$ which is also a Lie subalgebra of
  $(\cA,\PB)$, then $(\cB,\PB)$ is also a Poisson algebra; we say that it is a \emph{Poisson subalgebra} of
  $\cA$. Similarly, if $\cI$ is an ideal of $\cA$ which is also a Lie ideal of $(\cA,\PB)$ then $\cI$ is a
  \emph{Poisson ideal} of $\cA$ and $\cA/\cI$ is a Poisson algebra; we say that it is a \emph{quotient Poisson
  algebra} of $\cA$. The inclusion $\cB\to\cA$ and the projection $\cA\to\cA/\cI$ are then \emph{morphisms of
  Poisson algebras}, that is they are algebra as well as Lie algebra morphisms.
\end{example}

These examples will be used in what follows to construct, by Poisson reduction, commutative Poisson algebras from
deformations of a (not necessarily commutative) associative algebra, a notion which we first recall (see for
example \cite[Ch.\ 13]{PLV} for the commutative case).  Let $\cA$ be any associative algebra over $R$ and consider
$\cA[[\h]]$, the $R[[\h]]$-module of formal power series in some formal variable $\h$ with the elements of $\cA$ as
coefficients. By definition, any element $A$ of $\cA[[\h]]$ can be written in a unique way as $A=a_0+\h
a_1+\h^2a_2+\cdots,$ where $a_i\in\cA$ for all~$i$.

\begin{defn}\label{def:defo}
Suppose that $\cA[[\h]]$ is equipped with the structure of an associative algebra over~$R[[\h]]$, with product
denoted by~$\star$. Then $(\cA[[\h]],\star)$ (or more simply $\cA[[\h]]$) is said to be a \emph{(formal)
deformation} of $\cA$ if for any $a,b\in \cA$, $a\star b=ab+\cO(\h)$, i.e., $a\star b-ab\in\h \cA[[\h]]$.
\end{defn}
Said differently, the latter condition states that under the natural identification of $\cA$ with
$\cA[[\h]]/\h\cA[[\h]]$ the canonical projection $p:(\cA[[\h]],\star)\to (\cA,\cdot)$ is a morphism of
algebras. One naturally views $p$ as evaluation at $\h=0$.

\begin{example}\label{exa:moyal}
  A first classical example is the Moyal product which yields a non-trivial deformation of the algebra of functions
  on any symplectic manifold $(M,\omega)$. If we denote by $\Gamma$ the inverse to $\omega$, then $\Gamma$ is a
  Poisson structure on $M$ and the Moyal product of $f,g\in C^\infty(M)$ is given by
  \begin{equation*}
     f\star g=m\circ e^{\h\Gamma/2}(f\otimes g)\;,
  \end{equation*}
  where $m$ denotes the usual multiplication in $C^\infty(M)$.
\end{example}

\begin{example}\label{exa:Lie}
  A second classical example is the standard deformation of the algebra $\Sym\fg$ of polynomial functions on a Lie
  algebra $\fg$. If we denote by $T^\bullet\fg$ the tensor algebra of $\fg$ and by $\cI$ the two-sided ideal of
  $T^\bullet\fg[[\h]]$ generated by all $x\tns y-y\tns x-\h[x,y]$ with $x,y\in\fg$, then by the
  Poincaré-Birkhoff-Witt Theorem $T^\bullet\fg[[\h]]/\cI\simeq\Sym\fg[[\h]]$ making $\Sym\fg[[\h]]$ with the
  transported product $\star$ into a deformation of $\Sym\fg$.
\end{example}

See \cite{BFFLS1,BFFLS2} for more information on these examples and for the relevance of deformation theory to
quantisation. Notice that in both of these examples the associative algebra $\cA$, which is deformed, is
commutative.

The commutator in $(\cA[[\h]],\star)$ is denoted by $\LB_\star$: $\lb{A,B}_\star:=A\star B-B\star A$ for
$A,B\in\cA[[\h]]$.  We also introduce for all $i\in\bbZ_{>0}$, $R$-bilinear maps
$(\cdot,\cdot)_i,\PB_i:\cA\times\cA\to\cA$ by setting, for all $a,b\in\cA\subset\cA[[\h]]$,
\begin{align}
  a\star b&=ab+\h(a,b)_1+\h^2(a,b)_2+\cdots,\label{eq:prod_exp}\\
  \lb{a,b}_\star&=\lb{a,b}+\h\pb{a,b}_1+\h^2\pb{a,b}_2+\cdots,\label{eq:comm_exp}
\end{align}
where the values of the leading terms follow from $p(a\star b)=ab$ and $p(\lb{a,b}_\star)=ab-ba=\lb{a,b}$. Of
course, $\pb{a,b}_i=(a,b)_i-(b,a)_i$ for all $i$ and all $a,b\in\cA$.

\subsection{Commutative Poisson algebras from deformations}

It is well-known that when the associative algebra $\cA$ is commutative, $\cA$ inherits from any deformation
$(\cA[[\h]],\star)$ of $\cA$ a  Poisson bracket (see for example 
\cite{BFFLS1,PLV}). Recall that this
Poisson bracket is classically defined for $a,b\in \cA\subset\cA[[\h]]$ by
\begin{equation}\label{eq:classical_poisson}
  \pb{a,b}=\lim_{\h\to0}\frac{a\star b-b\star a}\h\;,
\end{equation}
and that the fact that it is a Poisson bracket follows from the associativity of $\star$.
In order to anti\-cipate the construction of the Poisson bracket in the noncommutative case, we first reformulate
the construction of the Poisson bracket~\eqref{eq:classical_poisson} in a different, more abstract way. In view of
Example~\ref{exa:commutator}, $(\cA[[\h]],\LB_\star)$ is a Poisson algebra over $R[[\h]]$. Since $\cA$ is
commutative, the commutator $\LB_\star$ takes values in $\h\cA[[\h]]$ and we can also consider on $\cA[[\h]]$ its
rescaling, defined for $A,B\in\cA[[\h]]$ by
\begin{equation}\label{eq:bracket_rescaled}
  \lb{A,B}_\h:=\frac1{\h}\lb{A,B}_\star=\frac{A\star B-B\star A}{\h}\in\cA[[\h]]\;.
\end{equation}
It is clear that this rescaling does not affect the Leibniz and Jacobi identities, so that $(\cA[[\h]],\LB_\h)$ is
also a Poisson algebra over the ring $R[[\h]]$. Since
\begin{equation*}
  \lb{\h\cA[[\h]],\cA[[\h]]}_\h=\lb{\cA[[\h]],\cA[[\h]]}_\star\subset\h\cA[[\h]]\;,
\end{equation*}
the (associative) ideal $\h\cA[[\h]]$ of $\cA[[\h]]$ is also a Lie ideal of $(\cA[[\h]],\LB_\h)$, hence is a
Poisson ideal of it. The quotient $\cA[[\h]]/\h\cA[[\h]]$ is therefore a Poisson algebra. Under the natural
identification of $\cA$ with $\cA[[\h]]/\h\cA[[\h]]$, we recover from \eqref{eq:bracket_rescaled} the Poisson
bracket \eqref{eq:classical_poisson} on $\cA$.

\begin{example}
  In the case of Example \ref{exa:moyal} the Poisson structure that one obtains is $\Gamma$. In the case
  of Example \ref{exa:Lie} one obtains the canonical Lie-Poisson structure on $\Sym\fg$ (see \cite[Ch.\ 7]{PLV}).
\end{example}

We now consider the more general case in which $\cA$ is not necessarily commutative. The center of $\cA$ is denoted
$Z(\cA)$. We suppose that $\cA[[\h]]$ is a deformation of $\cA$ and consider again the Poisson algebra
$(\cA[[\h]],\LB_\star)$. In this case we cannot define the rescaled bracket on $\cA[[\h]]$ as in
\eqref{eq:bracket_rescaled} because $\LB_\star$ does not take values in $\h\cA[[\h]]$ (in general). Let us
define $$\cH_\h:=Z(\cA)+\h\cA[[\h]]\;,$$ the $R[[\h]]$-submodule of $\cA[[\h]]$ generated by $Z(\cA)$ and
$\h\cA$. It consists of those elements of $\cA[[\h]]$ whose $\h$-independent term belongs to the center $Z(\cA)$
of~$\cA$.
\begin{lemma}\label{lma:ch_Poisson}
  $\cH_\h$ is a Poisson subalgebra of $(\cA[[\h]],\LB_\star)$. Moreover, the commutator $\LB_\star$, restricted to
  $\cH_\h$ takes values in $\h\cH_\h$, so that $(\cH_\h,\LB_\h)$ is also a Poisson algebra.
\end{lemma}
\begin{proof}
Since $Z(\cA)$ is a subalgebra of $\cA$, $\cH_\h$ is a subalgebra of $\cA[[\h]]$. But~$\cH_\h$ is also a Lie
subalgebra of $(\cA[[\h]],\LB_\star)$ because
\begin{equation*}
  \lb{\cH_\h,\cH_\h}_\star=\lb{Z(\cA)+\h\cA[[\h]],Z(\cA)+\h\cA[[\h]]}_\star\subset\h\cA[[\h]]\subset\cH_\h\;,
\end{equation*}
where we have used that $\lb{Z(\cA),Z(\cA)}_\star\subset\h\cA[[\h]]$.  It follows that $\cH_\h$ is a Poisson
subalgebra of $(\cA[[\h]],\LB_\star)$, which is the first statement, and that the bracket $\LB_\star$, restricted
to~$\cH_\h$ takes values in $\h\cA[[\h]]$. In order to prove the second statement, we need to show that the
restriction of $\LB_\star$ to $\cH_\h$ actually takes values in $\h\cH_\h$, i.e., that when $A,B\in\cH_\h$ then
$\lb{A,B}_\star\in\h\cH_\h$. Writing $A=a+\h a_1+\cO(\h^2)$ and $ B=b+\h b_1+\cO(\h^2)$ we have, using that $a$
and~$b$ belong to the center of~$\cA$ and using \eqref{eq:comm_exp}, that
\begin{equation}\label{eq:star_com}
  \lb{A,B}_\star=\lb{a,b}_\star+\h\lb{a,b_1}+\h\lb{a_1,b}+\cO\(\h^2\)=\h\pb{a,b}_1+\cO\(\h^2\)\;,
\end{equation}
so we need to show that $\pb{a,b}_1\in Z(\cA)$ for any $a,b\in Z(\cA)$. Let $c$ be any element of $\cA$. In view of
the Jacobi identity for $\LB_\star$ (which follows from the associativity of $\star$),
\begin{equation}\label{eq:Jacobi}
  \lb{\lb{a,b}_\star,c}_\star+  \lb{\lb{b,c}_\star,a}_\star+  \lb{\lb{c,a}_\star,b}_\star=0\;.
\end{equation}
Using \eqref{eq:star_com} and that $a\in Z(\cA)$, the first term reads
\begin{equation*}
  \lb{\lb{a,b}_\star,c}_\star=\h\lb{\pb{a,b}_1,c}_\star+\cO\(\h^2\)=\h(\pb{a,b}_1c-c\pb{a,b}_1)+\cO\(\h^2\)\;.
\end{equation*}
Now since $a,b\in Z(\cA)$,
\begin{equation*}
  \lb{\lb{b,c}_\star,a}_\star=\h\lb{\pb{b,c}_1,a}_\star+\cO\(\h^2\)=\h^2\pb{\pb{b,c}_1,a}_1+\cO\(\h^2\)=\cO\(\h^2\)\;,
\end{equation*}
and similarly $ \lb{\lb{c,a}_\star,b}_\star=\cO(\h^2)$. Substituted in \eqref{eq:Jacobi} we get
$\pb{a,b}_1c-c\pb{a,b}_1=0$ for all $c\in\cA$, which shows that $\pb{a,b}_1\in Z(\cA)$.
\end{proof}
We are now ready to construct the commutative Poisson algebra which is naturally associated with a deformation.
\begin{prop}\label{prop:poisson_from_deformation}
  Let $(\cA[[\h]],\star)$ be a deformation of an associative algebra $\cA$.
  \begin{enumerate}
    \item [(1)] $\h\cH_\h$ is a Poisson ideal of $(\cH_\h,\LB_\h)$.
    \item [(2)] The quotient algebra $\cH_\h/\h\cH_\h$ is a commutative Poisson algebra.
  \end{enumerate}
\end{prop}
\begin{proof}
$\h\cH_\h\subset\cH_\h$, which is an ideal of $\cH_\h$, is also a Lie ideal for the bracket $\LB_\h$ on $\cH_\h$
because
\begin{equation*}
  \lb{\h\cH_\h,\cH_\h}_\h=\h\lb{\cH_\h,\cH_\h}_\h\subset\h\cH_\h\;.
\end{equation*}
It follows that $\h\cH_\h$ is a Poisson ideal of $(\cH_\h,\LB_\h)$ and that the quotient algebra $\cH_\h/\h\cH_\h$
is a Poisson algebra. According to Lemma \ref{lma:ch_Poisson}, $A\star B-B\star A=\lb{A,B}_\star\in\h\cH_\h$ when
$A,B\in\cH_\h$, which shows that the Poisson algebra $(\cH_\h/\h\cH_\h,\LB_\h)$ is commutative.
\end{proof}

By construction, $\cH_\h/\h\cH_\h$ is a Poisson algebra over $R[[\h]]$, in particular, we can view
$\cH_\h/\h\cH_\h$ as a Poisson algebra over $R$, which is what we will do in what follows.  Moreover, we will write
elements of $\cH_\h/\h\cH_\h$ as pairs $(a,\cl b):=(a,b+Z(\cA))$ with $a\in Z(\cA)$ and $b\in\cA$, and identify
$\cH_\h/\h\cH_\h$ with
$$
  \Pi(\cA):=Z(\cA)\times \frac{\cA}{Z(\cA)}\;,
$$
which we call the \emph{Poisson algebra associated to the deformation} $(\cA[[h]],\star)$ of $\cA$.  Notice that
the quotient $\cA/Z(\cA)$ is not a quotient of algebras ($Z(\cA)$ is in general not an ideal of $\cA$) but of
$R$-modules. Under this identification and notation, the canonical projection $\pr:\cH_\h\to \Pi(\cA)$ is given by
$A=a_0+\h a_1+\h^2 a_2+\dots\mapsto (a_0,\cl{a_1})$.  The associative product and Poisson bracket on $\Pi(\cA)$ are
denoted by $\cdot$ and $\PB$ respectively. If we denote the unit of $\cA$ by $1$, then $(1,\cl0)$ is the unit of
$\Pi(\cA)$. By construction, we have the following corollary of Proposition \ref{prop:poisson_from_deformation}:
\begin{prop}\label{prp:pr_poisson}
  $\pr:(\cH_\h,\LB_\h)\to (\Pi(\cA),\PB)$ is a surjective morphism of Poisson algebras.\qed
\end{prop}
Explicit formulas for $\cdot$ and $\PB$ are given in the following proposition:
\begin{prop}\label{prp:PiA_products}
Let $(a,\cl{a_1})$, $(b,\cl{b_1})\in\Pi(\cA)$. Then
\begin{equation}\label{eq:prod}
  (a,\cl{a_1})\cdot(b,\cl{b_1})=\(ab,\cl{ab_1+a_1b+(a,b)_1}\)\;,
\end{equation}
and
\begin{align}\label{eq:PB}
  \pb{(a,\cl{a_1}),(b,\cl{b_1})}=\(\pb{a,b}_1,\cl{\pb{a,b}_2+\pb{a_1,b}_1+\pb{a,b_1}_1+\lb{a_1,b_1}}\)\;.
\end{align}
\end{prop}
\begin{proof}
  Since $\pr$ is surjective, we can write $(a,\cl{a_1})=\pr(A)$ and $(b,\cl{b_1})=\pr(B)$, where $A=a+\h
  a_1+\h^2a_2+\cdots$ and $B=b+\h b_1+\h^2b_2+\cdots$ belong to $\cH_\h$. Then, using Proposition
  \ref{prp:pr_poisson},
  \begin{eqnarray*}
    (a,\cl{a_1})\cdot(b,\cl{b_1})
    &=&  \pr(A)\cdot\pr(B)=\pr(A\star B)= \pr((a+\h a_1)\star(b+\h b_1))\\
    &\stackrel{\eqref{eq:prod_exp}}= &\pr(ab+\h(ab_1+a_1b+(a,b)_1))=\(ab,\cl{ab_1+a_1b+(a,b)_1}\)\;,
\end{eqnarray*}
and
\begin{eqnarray*}
  \pb{(a,\cl{a_1}),(b,\cl{b_1})}
  &=& \pr\lb{a+\h a_1+\h^2 a_2,b+\h b_1+\h^2b_2}_\h\\
  &\stackrel{\eqref{eq:comm_exp}}=&\pr(\pb{a,b}_1+\h(\pb{a,b}_2+\pb{a_1,b}_1+\pb{a,b_1}_1+\lb{a_1,b_1}))\\
  &=&\(\pb{a,b}_1,\cl{\pb{a,b}_2+\pb{a_1,b}_1+\pb{a,b_1}_1+\lb{a_1,b_1}}\)\;.
\end{eqnarray*}
\end{proof}
\begin{remark}\label{exa:class_poisson}
When $\cA$ is commutative, $Z(\cA)=\cA$ and $\cH_\h=\cA[[\h]]$ so that, as Poisson algebras,
$$
  \Pi(\cA)\simeq\cH_\h/\h\cH_\h\simeq\cA[[\h]]/\h\cA[[\h]]\simeq\cA\;,
$$
and we recover the (commutative) Poisson algebra constructed in the commutative case, with Poisson bracket
$\PB=\PB_1$.
%
\end{remark}

\begin{example}\label{exa:trivial}
  Let $\cA$ be any associative algebra over $R$ and consider the trivial deformation of $\cA$: the product $\star$
  on $\cA[[\h]]$ is the $R[[\h]]$-linear extension of the product on $\cA$. Then $(a,b)_i=\pb{a,b}_i=0$ for
  $a,b\in\cA$ and $i\geqslant1$, since $a\star b=ab$ for all $a,b\in\cA$. Hence, the formulas \eqref{eq:prod} and
  \eqref{eq:PB} for the product and Poisson bracket on $\Pi(\cA)$ become
  \begin{equation}\label{eq:example}
      \(a,\cl{a_1}\)\cdot\(b,\cl{b_1}\)= \(ab,\cl{ab_1+a_1b}\)\;,\quad\hbox{and}\quad
     \pb{(a,\cl{a_1}),\(b,\cl{b_1}\)}=\(0,\cl{\lb{a_1,b_1}}\)\;.
  \end{equation}
\end{example}
It follows that for any associative algebra $\cA$, \eqref{eq:example} defines the structure of a commutative
Poisson algebra on $\Pi(\cA)$. The Poisson bracket in \eqref{eq:example} is trivial if and only if $\cA$ is 2-step
nilpotent, $\lb{\lb{\cA,\cA},\cA}=0$.

In order to define the $\Pi(\cA)$-Poisson module structure on $\cA$, we will need a slightly larger Poisson
algebra, which is non necessarily commutative, given by the following proposition. Its proof is very similar to the
proof of Proposition \ref{prop:poisson_from_deformation}.
\begin{prop}\label{prp:big_Poisson}
  Let $(\cA[[\h]],\star)$ be a deformation of an associative algebra $\cA$.
  \begin{enumerate}
    \item [(1)] $\h^2\cA[[\h]]$ is a Poisson ideal of $(\cH_\h,\LB_\h)$.
    \item [(2)] The quotient algebra $\cH_\h/\h^2\cA[[\h]]\simeq Z(\cA)\times\cA$ is a Poisson algebra.\qed
  \end{enumerate}
\end{prop}
By a similar computation as in the proof of Proposition \ref{prp:PiA_products}, the product and the Poisson bracket
of $(a,a_1)$, $(b,b_1)\in Z(\cA)\times {\cA}$, again denoted by $\cdot$ and $\PB$, take the following form:
\begin{align}\label{eq:2prod_ZAA}
  (a,a_1)\cdot(b,b_1)&=(ab,ab_1+a_1b+(a,b)_1)\;,\nonumber\\
  \pb{(a,a_1),(b,b_1)}&=(\pb{a,b}_1,\pb{a,b}_2+\pb{a_1,b}_1+\pb{a,b_1}_1+\lb{a_1,b_1})\;.
\end{align}
It is clear that, alternatively, the commutative Poisson algebra $\Pi(\cA)$ can be constructed as a quotient of the
Poisson algebra $(Z(\cA)\times{\cA},\PB)$ by considering the Poisson ideal $\set0\times Z(\cA)$ of
$Z(\cA)\times{\cA}$.

The Poisson algebra $Z(\cA)\times\cA$ is not commutative when $\cA$ is not commutative. According to \cite{FL}, the
Poisson bracket of any prime Poisson algebra that is not commutative is a multiple of the commutator $a\cdot
b-b\cdot a$, where $\cdot$ denotes the (associative) product of the Poisson algebra. This result does not apply to
$Z(\cA)\times\cA$ because it is not prime. In fact, in this case the Poisson bracket is not a multiple of the
commutator.

\subsection{Poisson modules from deformations}
We first recall the definition of a Poisson module over a (not necessarily commutative) Poisson algebra.
\begin{defn}
  Let $(\cA,\cdot,\PB)$ be a Poisson algebra over $R$ and let $M$ be an $R$-module. Then~$M$ is said to be a
  $\cA$-\emph{Poisson module} (or Poisson module \emph{over} $\cA$ or \emph{over} $(\cA,\PB)$) when $M$ is both a
  $(\cA,\cdot)$-bimodule and a $(\cA,\PB)$-Lie module, satisfying the following derivation properties: for all
  $a,b\in\cA$ and $m\in M$,
  \begin{align}
    \pbm a{b\cdot m}&=\pb{a,b}\cdot m+b\cdot \pbm am\;,\label{eq:Leibniz_module_1}\\
    \pbm a{m\cdot b}&=m\cdot \pb{a,b}+\pbm am\cdot b\;,\label{eq:Leibniz_module_2}\\
     \pbm{a\cdot b}m&=a\cdot\pbm bm+\pbm am\cdot b\;.\label{eq:Leibniz_module_3}
  \end{align}
\end{defn}
In the above formulas, the three actions of $\cA$ on $M$ have been written $a\cdot m$, $m\cdot a$ and $\pbm am$ for
$a\in\cA$ and $m\in M$. In this notation, the fact that $M$ is a $\cA$-bimodule (respectively a $(\cA,\PB)$-Lie
module), takes the form
\begin{gather}
  a\cdot(b\cdot m)=(a\cdot b)\cdot m\;,\qquad (m\cdot a)\cdot b=m\cdot(a\cdot b)\;,\qquad
  a\cdot(m\cdot b)=(a\cdot m)\cdot b\;,\label{eq:def_poisson_module_1}\\
  \pbm{\pb{a,b}}m=\pbm a{\pbm bm}-\pbm b{\pbm
    am}\;,\label{eq:def_poisson_module_2}
\end{gather}
for $a,b\in\cA$ and $m\in M$. When $\cA$ is commutative and the left and right actions of $\cA$ on $M$ coincide,
\eqref{eq:Leibniz_module_1} and \eqref{eq:Leibniz_module_2} are equivalent, just like the first two conditions in
\eqref{eq:def_poisson_module_1}. The properties \eqref{eq:def_poisson_module_1} and \eqref{eq:def_poisson_module_2}
are similar to the associativity and Jacobi identity in $\cA$, while the properties \eqref{eq:Leibniz_module_1} --
\eqref{eq:Leibniz_module_3} are similar to the Leibniz identity in $\cA$. In the example that follows, they are
exactly these properties.

\begin{example}\label{exa:algebra=module}
  It is well-known that every associative algebra is a bimodule over itself and that every Lie algebra is a Lie
  module over itself, in both cases in a natural way. When $(\cA,\PB)$ is a Poisson algebra this leads to a natural
  $\cA$-bimodule structure on $\cA$, given by left and right multiplication, as well as a $(\cA,\PB)$-Lie module
  structure, given by taking the Poisson bracket. Then $\PBm=\PB$ and each one of the properties
  \eqref{eq:Leibniz_module_1} -- \eqref{eq:Leibniz_module_3} is equivalent to the Leibniz identity in
  $(\cA,\PB)$. It follows that every Poisson algebra is in a natural way a Poisson module over itself.
\end{example}

\begin{prop}\label{prp:poisson_module}
  Let $(\cA[[\h]],\star)$ be a deformation of an associative algebra $\cA$. Consider the Poisson algebra
  $(Z(\cA)\times\cA,\cdot,\PB)$ (Proposition \ref{prp:big_Poisson}).
  \begin{itemize}
  \item[(1)] $\cA$ is a Poisson module over $(Z(\cA)\times\cA,\cdot,\PB)$.
  \item[(2)] $\cA$ is a Poisson module over $(\Pi(\cA),\cdot,\PB)$, with actions given for
    $a\in Z(\cA)$ and $a_1,b\in \cA$~by
    \begin{equation}\label{eq:final_module}
      (a,\cl{a_1})\cdot b=ab\;,\quad
      b\cdot(a,\cl{a_1})=ba=ab\;,\quad
      \pbm{(a,\cl{a_1})}b=\pb{a,b}_1+\lb{a_1,b}\;.
  \end{equation}
  \end{itemize}
\end{prop}
\begin{proof}
It is clear from Example \ref{exa:algebra=module} that the Poisson algebra $(Z(\cA)\times\cA,\PB)$ is a Poisson
module over itself.  The formulas for the product and bracket are given by~\eqref{eq:2prod_ZAA}. Let
$(a,a_1),(0,b)\in Z(\cA)\times\cA$. Then ~\eqref{eq:2prod_ZAA} specializes to
\begin{equation}\label{eq:sub_module}
  (a,a_1)\cdot(0,b)=(0,ab)\;,\quad
  (0,b)\cdot(a,a_1)=(0,ba)\;,\quad
  \pb{(a,a_1),(0,b)}=\(0,\pb{a,b}_1+\lb{a_1,b}\)\;,
\end{equation}
so that $\cdot$ and $\PB$ can be restricted to $\set0\times\cA$, making the latter into a Poisson module over
$(Z(\cA)\times\cA,\PB)$.  Under the identification $\set0\times\cA\simeq\cA$ we get (1). We use these formulas to
prove~(2), where we need to show that the restriction to the Poisson ideal $\set0\times Z(\cA)$ of both actions of
$Z(\cA)\times\cA$ on $\cA$ is trivial. Let $a\in Z(\cA)$ and $b\in\cA$. Then~\eqref{eq:sub_module} becomes
\begin{equation}\label{eq:sub_module_1}
  (0,a)\cdot(0,b)=(0,0)=(0,b)\cdot(0,a)\;,\hbox{ and }
  \pb{(0,a),(0,b)}=(0,\lb{a,b})=(0,0)\;,
\end{equation}
where we have used that $a\in Z(\cA)$. In terms of the notation that we have introduced for the elements of
$\Pi(\cA)$, upon identifying $\set0\times\cA$ with $\cA$ and writing $\PBm$ for the induced Lie action, we get
\eqref{eq:final_module} from \eqref{eq:sub_module}.
\end{proof}

We stress that the Poisson module $\cA$ that we have constructed is a Poisson module over the \emph{commutative}
Poisson algebra $\Pi(\cA)$. Notice also that the left and right actions of $Z(\cA)\times\cA$, and hence of
$\Pi(\cA)$, on $\cA$ are the same (see \eqref{eq:final_module} and \eqref{eq:sub_module_1}).

\begin{example}
As we have already pointed out, when $\cA$ is commutative, \hbox{$\Pi(\cA)\simeq\cA$} as a Poisson algebra, where
both algebras have the rescaled commutator $\LB_\h$ as Poisson bracket. Under this identification, the
$\cA$-Poisson module structure of $\cA$ constructed in Proposition \ref{prp:poisson_module} is precisely the
canonical Poisson module structure of $\cA$ as a Poisson module over itself (cfr.\ Example
\ref{exa:algebra=module}).
\end{example}


%
\subsection{Hamiltonian derivations from deformations}\label{par:heisenberg}
We now show how any element $\bH$ of $\Pi(\cA)$ leads to a \emph{derivation} $\p_\bH$ on $\cA$, i.e., a linear map
$\cA\to\cA$ satisfying the Leibniz identity. We call them \emph{Hamiltonian derivations} in analogy with the
terminology used in the classical case, i.e., when~$\cA$ is commutative, so that $\Pi(\cA)\simeq\cA$. For $a\in\cA$
, we define $\p_\bH a:=\pbm{\bH}a$. We show in the following proposition that $\p_\bH$ is a derivation of $\cA$.

\begin{prop}\label{prp:derivation}
  $\p_\bH$ is a derivation of $\cA$.
\end{prop}
\begin{proof}
Recall from \eqref{eq:final_module} that
\begin{equation}\label{eq:final_module_recall}
  \pbm{\bH}a=\pb{H_0,a}_1+\lb{H_1,a}\;.
\end{equation}
It is clear that $\p_\bH$ is a linear map, so we only need to establish that $\pBm{\bH}$ satisfies the Leibniz
identity. Now $\lB{H_1}$ satisfies the Leibniz identity, so according to \eqref{eq:final_module_recall} we only
need to check that for any $a,b\in\cA$, $\pb{H_0,ab}_1=a \pb{H_0,b}_1+\pb{H_0,a}_1b$. This follows by comparing the
following two expressions for $\lb{H_0,ab}_\star$:
\begin{align*}
  \lb{H_0,ab}_\star&=\lb{H_0,ab}+\h\pb{H_0,ab}_1+\cO\(\h^2\)=a\lb{H_0,b}+\lb{H_0,a}b+\h\pb{H_0,ab}_1+\cO\(\h^2\)\;,\\
  \lb{H_0,ab}_\star&=a\lb{H_0,b}_\star+\lb{H_0,a}_\star b=a\lb{H_0,b}+a\h\pb{H_0,b}_1+\lb{H_0,a}b+\h\pb{H_0,a}_1b+
  \cO\(\h^2\)\;.
\end{align*}
\end{proof}

The Hamiltonian derivation $\p_\bH$ on $\cA$ should not be confused with the Hamiltonian derivation $\p'_\bH$ on
$\Pi(\cA)$, which is defined for $\bF\in\Pi(\cA)$ by $\p'_\bH\bF:=\pb{\bH,\bF}$, where we recall that the latter
Poisson bracket is explicitly given by \eqref{eq:PB}. When $\cA$ is commutative, $\p_\bH$ and $\p_\bH'$ are both
derivations of $\cA$ (under the obvious identifications) and they coincide. When $\cA$ is not commutative, these
derivations are defined on different algebras and though they are related, none of the two determines the other
one.

The following proposition generalizes a well-known property from the commutative case:
\begin{prop}\label{prp:invol}
  Suppose that $\bF=(F_0,\cl{F_1}),\bG=(G_0,\cl{G_1})\in\Pi(\cA)$. Then $\lb{\p_\bF,\p_\bG}=\p_{\pb{\bF,\bG}}$.  In
  particular, if $\bF$ and $\bG$ are in involution, $\pb{\bF,\bG}=0$, their associated derivations $\p_\bF$
  and~$\p_\bG$ of $\cA$ commute and both $F_0$ and $G_0$ are first integrals of them.
\end{prop}
\begin{proof}
For the first statement, we need to prove that for any $a\in\cA$,
\begin{equation*}
  \p_\bF(\p_\bG a)-\p_\bG(\p_\bF a)=\p_{\pb{\bF,\bG}}a\;,
\end{equation*}
i.e., that
\begin{equation*}
  \pbm{\bF}{\pbm{\bG}a}-\pbm{\bG}{\pbm{\bF}a}=\pbm{\pb{\bF,\bG}}a\;.
\end{equation*}
This is precisely the property that $\cA$ is a Lie module over $(\Pi(\cA),\PB)$ (see
\eqref{eq:def_poisson_module_2}). For the last statement, it remains to be shown that $G_0$ is a first integral of
$\p_\bF$, i.e., that $\p_\bF G_0=0$. According to \eqref{eq:final_module_recall},
$\p_\bF G_0=\pbm\bF{G_0}=\pb{F_0,G_0}_1+\lb{F_1,G_0}$ and both terms in the sum are zero, the first one because $\bF$ and $\bG$
are in involution (see \eqref{eq:PB}) and the second one because $G_0\in Z(\cA)$.
\end{proof}

The proposition also shows that the commutator of two Hamiltonian derivations of $\cA$ is a Hamiltonian derivation
of $\cA$ and the Poisson algebra $\Pi(\cA)$ provides a tool to compute a corresponding Hamiltonian. Notice that,
contrarily to the classical case where $\Pi(\cA)\simeq\cA$, in the general case Hamiltonians are pairs of
$Z(\cA)\times\frac{\cA}{Z(\cA)}$, while first integrals are elements of $\cA$, so they live on different spaces.

%
The derivation $\p_\bH$ of $\cA$, given by \eqref{eq:final_module_recall}, is naturally obtained as the limit of a
Heisenberg derivation of $\cA[[\h]]$, an observation which has been our original motivation for the construction of
the Poisson algebra $\Pi(\cA)$ and the Poisson module~$\cA$ over it. Indeed, let $H\in\cH_\h$ be any element with
$H=H_0+\h H_1+\h^2 H_2+\cdots$, and denote as before $\bH=(H_0,\cl{H_1})$. The \emph{Heisenberg derivation}
associated with $H$, which we denote by $\delta_H$, is by definition the derivation of $\cA[[\h]]$, given for
$a\in\cA\subset\cA[[\h]]$ by
\begin{equation*}
  \delta_Ha=\frac1{\h}\lb{H,a}_\star\;.
\end{equation*}
We get, using \eqref{eq:comm_exp} and \eqref{eq:final_module_recall},
\begin{equation*}
  \delta_Ha=\frac1{\h}\lb{H_0+\h H_1+\cdots,a}_\star=\pb{H_0,a}_1+\lb{H_1,a}+\cO(\h)=\p_\bH a+\cO(\h)\;,
\end{equation*}
so that in the limit $\h\to0$ we get indeed the Hamiltonian derivation $\p_\bH$ of~$\cA$.

\begin{example}
When $\cA$ is commutative one recovers the well-known fact that the limit of the Heisenberg derivation associated
with $H_0$, where $\bH=\(H_0,\cl{H_1}\)=(H_0, \cl0)$, is the Hamiltonian derivation associated with $H_0$ by means
of the Poisson structure, associated with the deformation since $\p_\bH a=\pb{H_0,a}_1+\lb{H_1,a}=\pb{H_0,a}$,
since in the commutative case the commutator vanishes and $\PB_1$ is the Poisson bracket (see
Remark~\ref{exa:class_poisson}).
\end{example}


\subsection{The reduced Poisson algebra}\label{sec:reduced}

Let $(\cA,\star)$ be a deformation of an associative algebra $\cA$ and recall that $\Pi(\cA)=Z(\cA)\times
({\cA}/{Z(\cA)})$ is the associated Poisson algebra. Explicit formulas for the product and bracket on $\Pi(\cA)$
are given in \eqref{eq:prod} and \eqref{eq:PB}. It is clear from these formulas that $Z(\cA)\times\set{\cl{0}}$ is
in general neither a subalgebra nor a Lie subalgebra of $\Pi(\cA)$. Yet, we show that, by reduction, $Z(\cA)$ is
also a (commutative) Poisson algebra. In fact, $\set0\times ({\cA}/{Z(\cA)})$ is both an ideal and a Lie ideal of
$\Pi(\cA)$, since according to \eqref{eq:prod} and \eqref{eq:PB},
\begin{equation*}
  (0,\cl{a_1})\cdot\(b,\cl{b_1}\)=\(0,\cl{a_1b}\)\;,\quad
  \pb{(0,\cl{a_1}),\(b,\cl{b_1}\)}  =\(0,\cl{\pb{a_1,b}_1+\lb{a_1,b_1}}\)\;.
\end{equation*}
It follows that $\set0\times (\cA/Z(\cA))$ is a Poisson ideal of $\Pi(\cA)$, so that
$\Pi(\cA)/(\set0\times(\cA/Z(\cA))) \simeq Z(\cA)$ is a Poisson algebra, where the latter isomorphism is according
to \eqref{eq:prod} not just an isomorphism of modules but of (associative) algebras. The induced Poisson structure
on $Z(\cA)$ is given for $a,b\in Z(\cA)$ by $\pb{a,b}_1$ since according to \eqref{eq:PB},
\begin{equation*}
  \pb{\(a,\cl0\),\(b,\cl0\)}=  \(\pb{a,b}_1,\cl{\pb{a,b}_2}\)\;.
\end{equation*}
We call $(Z(\cA),\PB_1)$ the \emph{reduced Poisson algebra (associated to the deformation)}. Notice that it follows
that $\PB_1$ satisfies the Jacobi identity when restricted to $Z(\cA)$, though in general $\PB_1$ does not satisfy
the Jacobi identity on $\cA$.

\begin{example}
  According to Example \ref{exa:trivial}, the reduced Poisson structure associated with a trivial deformation is trivial.
\end{example}

\subsection{Functoriality}

In deformation theory, the natural notion of isomorphism is the one of equivalence. We will show in this subsection
that equivalent deformations of an algebra lead to isomorphic Poisson algebras. We first recall the definition of
equivalence of deformations (see for example \cite[Ch.\ 13]{PLV}).

\begin{defn}
  Two deformations $(\cA[[\h]],\star)$ and $(\cA[[\h]],\star')$ of an associative algebra $\cA$ are said to be
  \emph{equivalent} if there exists a morphism of $R[[\h]]$-algebras $F:(\cA[[\h]],\star)\to(\cA[[\h]],\star')$
  such that $F(a)=a+\cO(\h)$ for all $a\in\cA$. Then $F$ is called an \emph{equivalence (of deformations)}.
\end{defn}

Notice that $F$ is automatically an isomorphism and that $F$ can be viewed as a deformation of the identity map on
$\cA[[\h]]$. More precisely, expanding $F(a)$ for all $a\in\cA$ as a formal power series in $\h$,
\begin{equation*}
  F(a)=a+\h F_1(a)+\h^2 F_2(a)+\cdots\;,
\end{equation*}
we get $R$-linear maps $F_i:\cA\to\cA$ and we can write
\begin{equation}\label{eq:F_exp}
  F=\Id_{\cA[[\h]]}+\h F_1+\h^2 F_2+\cdots
\end{equation}
where, by a slight abuse of notation, $F_i$ stands for the $R[[\h]]$-linear extension of $F_i$ to a morphism
$F_i:\cA[[\h]]\to\cA[[\h]]$. Formula \eqref{eq:F_exp} is convenient for computations. One computes for example
easily from it that $F^{-1}=\Id_{\cA[[\h]]}-\h F_1+\h^2(F_1^2-F_2)+\cdots$.

The following proposition shows that, under some condition which is automatically satisfied for equivalences, an
algebra homomorphism between deformations of two algebras induces a Poisson morphism between the two associated
Poisson algebras.

\begin{prop}
  Let $\cA$ and $\cA'$ be two associative algebras with deformations $(\cA[[\h]],\star)$ and $(\cA'[[\h]],\star')$,
  and let $F=F_0+\h F_1+\h^2F_2+\cdots:\cA[[\h]] \to\cA'[[\h]]$ be a morphism of $R[[\h]]$-algebras. If
  $F_0(Z(\cA))\subset Z(\cA')$ then $F$ restricts to a morphism of Poisson algebras
  $F:(\cH_\h,\LB_\h)\to(\cH_\h',\LB_\h')$ and induces a map $F_\Pi;\Pi(\cA)\to\Pi(\cA')$. Moreover, $F_\Pi$ and
  the map induced by it between the reduced Poisson algebras $(Z(\cA),\PB_1)$ and $(Z(\cA'),\PB_1')$ are
  morphisms of Poisson algebras.
\end{prop}
\begin{proof}
Suppose that $F_0(Z(\cA))\subset Z(\cA')$. Then
\begin{equation*}
  F(\cH_\h)=(F_0+\h F_1+\cdots)(Z(\cA)+\h\cA[[\h]])\subset Z(\cA')+\h\cA'[[\h]]=\cH_\h'\;.
\end{equation*}
It follows that $F$ can be restricted to a morphism $F:\cH_\h\to\cH_\h'$. Notice that since $F$ is a morphism of
$R[[\h]]$-algebras it also preserves the commutator, $F(\lb{A,B}_\star)=\lb{F(A),F(B)}_\star'$, and hence also the
rescaled bracket, $F(\lb{A,B}_\h)=\lb{F(A),F(B)}_\h'$, so that $F:(\cH_\h,\LB_\h)\to(\cH_\h',\LB_\h')$ is a
morphism of Poisson algebras. Consider the following diagram of Poisson algebras:
%
\begin{center}
  \begin{tikzcd}[row sep=6ex, column sep=12ex]
    (\cH_\h,\LB_\h)\arrow[swap]{d}{\pr}\arrow{r}{F}&(\cH'_\h,\LB'_\h)\arrow{d}{\prp}\\
    (\Pi(\cA),\PB)\arrow[swap]{r}{F_\Pi}&(\Pi(\cA'),\PB')
  \end{tikzcd}
\end{center}
%
As we just showed, $F$ is a morphism of Poisson algebras. According to Proposition \ref{prp:pr_poisson}, $\pr$ and
$\prp$ are also morphisms of Poisson algebras. Let us now consider $F_\Pi$; by surjectivity of $\pr$, if a map
$F_\Pi$ making the diagram commutative exists, it is unique. In fact, one establishes quite easily a formula for
$F_\Pi$: for $(a,\cl{a_1})\in\Pi(\cA)$,
\begin{align}
  F_\Pi(a,\cl{a_1})&=F_\Pi\pr(a+\h a_1+\cdots)  =\prp F(a+\h a_1+\cdots)\nonumber\\
  &=\prp(F_0(a)+\h(F_0(a_1)+F_1(a)))=\(F_0(a),\cl{F_0(a_1)+F_1(a)}\)\;.\label{eq:F_Pi}
\end{align}
This gives a formula for $F_\Pi$ and the above computation shows that it makes the diagram commutative. The
surjectivity of $\pr$ and the commutativity of the diagram imply that $F_\Pi$ is a morphism of Poisson
algebras. For example, to check that $F_\Pi$ is a morphism of Lie algebras, it suffices to check that $F_\Pi\pb{\pr
  A,\pr B}=\pb{F_\Pi\pr A,F_\Pi\pr B}'$ for all $A,B\in\cA[[\h]]$, which follows easily from the cited
properties:
\begin{align*}
  F_\Pi\pb{\pr A,\pr B}&=F_\Pi\pr\lb{A,B}_\h =\prp F\lb{A,B}_\h=\prp \lb{F(A),F(B)}_\h'\\
   &=\pb{\prp F(A),\prp F(B)}'=\pb{F_\Pi\pr A,F_\Pi\pr B}'\;.
\end{align*}
To finish to proof, we consider the corresponding reduced Poisson algebras (see Section
\ref{sec:reduced}). Consider the following diagram of Poisson algebras:
\begin{center}
  \begin{tikzcd}[row sep=6ex, column sep=12ex]
    (\Pi(\cA),\PB)\arrow[swap]{d}{}\arrow{r}{F_\Pi}&(\Pi(\cA'),\PB')\arrow{d}{}\\
    (Z(\cA),\PB_1)\arrow[swap]{r}{F_0}&(Z(\cA'),\PB'_1)
  \end{tikzcd}
\end{center}
%
The vertical arrows in this diagram are the reduction maps. The lower arrow has been labeled $F_0$, as it is just
the restriction of $F_0$ to the center of~$\cA$; it is obvious from \eqref{eq:F_Pi} that $F_0$ makes the diagram
commutative. Since $F_\Pi$ and the reduction maps are Poisson maps, the latter moreover being surjective, we may
conclude as above that $F_0$ is a Poisson map.
\end{proof}
When $F_0$ is surjective, the condition $F_0(Z(\cA))\subset Z(\cA')$ is automatically satisfied. In particular, the
proposition can be applied to equivalences of deformations and we get the following result:

\begin{cor}\label{cor:equivalent}
  Let $(\cA[[\h]],\star)$ and $(\cA[[\h]],\star')$ be equivalent deformations of an associative algebra~$\cA$. Then
  the Poisson algebras associated to these deformations are isomorphic, $(\Pi(\cA),\PB)\simeq (\Pi(\cA),\PB')$ ;
  the reduced Poisson algebras are also isomorphic, $(Z(\cA),\PB_1)\simeq (Z(\cA),\PB_1') $.\qed
\end{cor}

The converse is not true in general, as the Poisson algebras and reduced Poisson algebras only ``see'' the first
two terms of the deformation; replacing in a non-trivial deformation $\h$ by $\h^k$ leads to a non-trivial
deformation for which the first $k-1$ deformation terms are zero, hence leading to a trivial Poisson algebra
$\Pi(\cA)$ when $k>2$.

We show that under the hypothesis of Corollary \ref{cor:equivalent}, the Poisson module structure on $\cA$ which is
induced by the two deformations is the same, i.e., that for any $(a,\cl{a_1})\in\Pi(\cA)$ and $b$ in $\cA$, one has
$F_\Pi(a,\cl{a_1})\cdot' b=(a,\cl{a_1})\cdot b$ and $\pbm{F_\Pi(a,\cl{a_1})}b'=\pbm{(a,\cl{a_1})} b$, where
$\cdot'$ denotes the action of $\Pi(\cA)$ on $\cA$ coming from the deformation $(\cA[[\h]],\star')$. Now
$F_\Pi\(a,\cl{a_1}\)=\(a,\cl{F_1(a)+a_1}\)$, as follows from \eqref{eq:F_Pi}, so that according to
\eqref{eq:final_module} we need to show that
\begin{equation}\label{eq:equiv_mod}
  \(a,\cl{a_1+F_1(a)}\)\cdot' b=(a,\cl{a_1})\cdot b\;,\quad\hbox{and}\quad
  \pbm{\(a,\cl{a_1+F_1(a)}\)}b'=\pbm{(a,\cl{a_1})}b\;,
\end{equation}
for any $a\in Z(\cA)$ and $a_1,b\in\cA$. The first equality in \eqref{eq:equiv_mod} holds because both sides
evaluate to $ab$ (see \eqref{eq:final_module}). Proving the second equality amounts in view of the formulas
\eqref{eq:final_module} to showing that $\pb{a,b}_1=\pb{a,b}_1'+\lb{F_1(a),b}$ for all $a\in Z(\cA)$ and
$b\in\cA$. To prove the latter, we expand both sides of $\lb{F(a),F(b)}_\star'=F\lb{a,b}_\star$ using
\eqref{eq:F_Pi}. First,
\begin{align*}
  \lb{F(a),F(b)}_\star'&=\lb{a+\h F_1(a),b+\h F_1(b)}_\star'+\cO\(\h^2\)\\
    &=\lb{a,b}+\h(\pb{a,b}_1'+\lb{a,F_1(b)}+\lb{F_1(a),b})+\cO\(\h^2\)=\h(\pb{a,b}_1'+\lb{F_1(a),b})+\cO\(\h^2\)\;,
\end{align*}
where we have used twice that $a\in Z(\cA)$. Similarly,
\begin{align*}
  F\lb{a,b}_\star&=\lb{a,b}_\star+\h F_1\lb{a,b}_\star+\cO\(\h^2\)=\lb{a,b}+\h(\pb{a,b}_1+F_1\lb{a,b})+\cO\(\h^2\)
  =\h\pb{a,b}_1+\cO\(\h^2\)\;.
\end{align*}
Comparing these two results yields the required equality, proving the second equality in \eqref{eq:equiv_mod}.

\section{Free quotients of free algebras}\label{sec:ideals}
The associative algebra underlying a quantum system usually depends on a parameter, the Planck constant, which is
small and the algebra becomes commutative when the parameter is set to 0. These algebras can (under some
assumptions) naturally be viewed as deformations (in the sense of Definition \ref{def:defo}) of that commutative
algebra and one speaks of \emph{deformation quantisation}. We are interested here in the more general case in which
the limiting algebra is not necessarily commutative. Moreover, in view of the examples which we will treat, we may
be interested in other values of the parameter, such as roots of unity, and we may want to deal with several
parameters. We will describe in this section a natural and quite general setup in which we can view these more
general algebras as deformations, so that we can apply the techniques and results of the previous sections to them,
as we will do in Sections \ref{sec:ex1} and \ref{sec:ex2}. We denote by $\bbK$ any commutative field of
characteristic 0, which in the examples in the next sections will be taken equal to $\bbC$.

\subsection{Quantisation ideals and quantum algebras}

We first recall the notion of quantisation ideal and of quantum algebra which were first introduced in
\cite{AvM20}. Let $x_1,x_2,\dots$ be a (possible infinite, but at most countable) collection of independent
variables. We denote $\fA=\bbK\langle{x_1,\;x_2,\;\dots\rangle}$ the free associative (unital) $\bbK$-algebra on
these variables. Elements of $\fA$ are finite $\bbK$-linear combinations of words in $x_1,x_2,\dots,$ and the
product of two words is their concatenation. Assume that~$\fA$ is equipped with a derivation $\p:\fA\to\fA$.

\begin{defn}\label{def:quantum_algebra}
A two-sided ideal $\cI$ of $\fA$ is said to be a \emph{quantisation ideal} for $(\fA,\p)$ if it satisfies the
following two properties:
\begin{enumerate}
  \item[(1)] The ideal $\cI$ is $\p$-stable: $\p(\cI)\subset\cI;$
  \item[(2)] The quotient $\fA/\cI$ admits a basis $\cB$ of normally ordered monomials\footnote{Our convention is
  that a monomial in the generators $x_1,\,x_2,\,x_3,\dots$ is normally ordered when it is of the form
  $x_{m_1}^{n_1}x_{m_2}^{n_2}\dots x_{m_s}^{n_s}$ where $m_1,m_2,\cdots,m_s$ is strictly increasing, all $n_i$ are
  (usually positive) integers and $s\geqslant0$. Notice that it is not required that all such elements are in the
  basis~$\cB$, i.e.\ that $\cB$ is a PBW basis.} in $x_1,\;x_2,\;\dots$.
\end{enumerate}
 The quotient algebra $\fA/\cI$ is then said to be a \emph{quantum algebra} (over $\bbK$).
\end{defn}
The first condition implies that $\p$ descends to a derivation of $\fA/\cI$; we will come back to this in Section
\ref{sec:ex2}. In the absence of a derivation, we can still speak of a quantisation ideal and of a quantum algebra
by considering the trivial (zero) derivation of $\fA$; then (1) is automatically satisfied.

The generators of the quantisation ideals that define quantum algebras often depend on one or several parameters
$\q=(q_1,q_2,\dots)$, which can be thought of as being elements of the field~$\bbK$. Moreover, in all relevant
examples this dependency is rational. It is then natural to think of the family of quantum $\bbK$-algebras
$\fA/\cI$, depending on $\q$, as being the quantum $\bbK(\q)$-algebra $\cA_{\q}:=\fA(\q)/\cI_{\q}$, where~$\fA(\q)$
is a shorthand for $\bbK(\q)\langle{x_1,\;x_2,\;\dots\rangle}$ and~$\cI_{\q}$ stands for the ideal of $\fA(\q)$,
with the elements of $\cI$ being considered as elements of $\fA(\q)$. The basis $\cB$ of normally ordered monomials
is then a basis for $\cA_{\q}$ as a $\bbK(\q)$-module and will be simply referred to as a \emph{monomial basis} for
$\cA_{\q}$. The product and commutator of elements $A,B$ of~$\cA_\q$ will be denoted $AB$ and $\lb{A,B}$
respectively.
\begin{example}\label{exa:typical}
Let $\fA=\bbK\langle{x_1,\;x_2,\;\dots\;,x_N}\rangle$ be a free associative algebra as above and consider the ideal
$\cI_\q$ of $\fA(\q)$, generated by
\begin{equation*}
  \set{x_ix_j-q_{i,j}x_jx_i \mid 1\leqs j<i\leqs N}\;,
\end{equation*}
where $\q=(q_{i,j})_{1\leqs j<i\leqs N}$ are the parameters. Then it is clear that a monomial basis $\cB$ for
$\cA_\q=\fA(\q)/\cI_\q$ is given by
\begin{equation}\label{eq:basis}
   \cB=\set{x_1^{i_1}x_2^{i_2}\cdots x_N^{i_N}  \mid i_1,\;i_2,\;\dots,i_N\in\bbN}\;,
\end{equation}
where we have used the same notation $x_i$ for the generators of $\fA$ as for their images in $\cA_\q$. Notice that
in $\cA_\q$ we have $x_ix_j=q_{i,j}x_jx_i$ for $i>j$, which can be written equivalently as
$\lb{x_i,x_j}=(q_{i,j}-1)x_jx_i$, in particular the commutator in $\cA_\q$ of any two generators $x_i$ is a multiple
of their product; more generally, the commutator in $\cA_\q$ of two monomials in $x_1,\dots,x_N$ is a multiple of
their product.  A basis of the form \eqref{eq:basis} is called a \emph{PBW basis}.
\end{example}

\begin{example}
  With $\fA$ as in the previous example, consider now the ideal $\cI_\q$ of $\fA(\q)$, generated by
\begin{equation*}
  \set{x_ix_j-x_jx_i-q_{i,j} \mid 1\leqs j<i\leqs N}\;,
\end{equation*}
where $\q=(q_{i,j})_{1\leqs j<i\leqs N}$ are again the parameters. Then $\cB$, as given by \eqref{eq:basis}, is
again a PBW basis for $\cA_\q$. However, in this case the commutator in $\cA_\q$ of two monomials in $x_1,\dots,x_N$
is in general not a multiple of their product. Both examples can of course be combined to produce more general
quantisation ideals and quantum algebras with a PBW basis. See \cite{Levandovskyy} for a very large class of ideals
whose associated quotient algebra has a basis of normally ordered monomials, making it into a quantum algebra.
\end{example}

\subsection{Free quotients and deformations}\label{par:quotients_and_deformations}
We now show how a quantum algebra which depends on one or several parameters can be turned into a (formal)
deformation of an algebra (Definition~\ref{def:defo}). This will allow us to apply the results and methods of
Section \ref{sec:poisson}. For clarity, and in view of the examples which we will treat, let us assume that there
is only one parameter, $\q=q$; see Remark~\ref{rem:several_parameters} below for the case of several parameters. We
would like to specialize $q$ to a value $q_0$, but such that~$\cB$ is still a basis after specialization. Notice
that since $\cB$ is a basis for $\cA_q=\fA(q)/\cI_q$, we can write for any $j>i$ the product $x_jx_i$ as a finite
linear combination of elements of $\cB$, with as coefficients rational fractions in $q$.  When none of these
fractions has a pole at $q_0$, we will say that $q_0$ is a \emph{regular} value of $\cB$. When $q_0$ is a regular
value of $\cB$, the expression for $x_jx_i$ in terms of $\cB$ can be evaluated at $q_0$. This shows that $\cB$
remains a (monomial) basis after specialization.  In fact, these expressions for $x_jx_i$ with $j>i$ can be taken
as generators of $\cI_q$; they will be used as such in what follows.  Suppose that $q_0$ is a regular value of
$\cB$ and set $q(\h)=q_0+c_1\h+c_2\h^2+\cdots$, any polynomial or formal power series in~$\h$ with $c_1\neq0$; in
practice we will take $q(\h)=q_0+\h$, see Remark~\ref{rem:rescaling} below.  Elements of $\bbK(\q)$ which do not
have a pole at $q_0$ are then formal power series in $\h$.  We consider in
$\fA[[\h]]:=\bbK[[\h]]\langle{x_1,\;x_2,\;\dots\rangle}$ the closed ideal~$\cI_{q(\h)}$ generated by the generators
of $\cI_q$ in which every occurrence of $q$ has been replaced by $q(\h)$, and denote
$\cA_\h:=\fA[[\h]]/\cI_{q(\h)}$. We consider in particular the $\bbK$-algebra $\cA:=\cA_{q_0}:=\fA/\cI_{q_0}$. Then
$\cB$ is a $\bbK$-basis of $\cA$ and is a $\bbK[[\h]]$-basis of $\cA_\h$. Notice that since $q_0$ is a regular
value of $\cB$, the elements of $\cI_{q(\h)}$ are formal power series in $\h$.

\goodbreak

\begin{prop}\label{prp:iso_to_formal}
  $\cA_\h$ is isomorphic, as a $\bbK[[\h]]$-module, to~$\cA[[\h]]$.
\end{prop}
\begin{proof}
We use the basis $\cB$ construct a surjective morphism $\fA[[\h]]\to\cA[[\h]]$ with kernel $\cI_{q(\h)}$. Let
$A=\sum_{i=0}^\infty a_i\h^i$ be any element of $\fA[[\h]]$, where $a_i\in\fA$ for all $i$. Since $\cB$ is a basis
of $\cA_\h$, each coefficient $a_i$ can be written uniquely as $a_i=\sum_j\alpha_{i,j}\h^j+r_i$, where the sum is
finite, all $\alpha_{i,j}$ belong to $\Span_\bbK\cB$ and $r_i\in\cI_{q(\h)}$. Summing up, we can write $A$ uniquely
as
\begin{equation}\label{eq:A_decomp}
  A=\sum_{i=0}^\infty a_i\h^i=\sum_{i,j}\alpha_{i,j}\h^{i+j}+\sum_{i=0}^\infty r_i\h^i
  =\sum_{k=0}^\infty\beta_k\h^k+\sum_{i=0}^\infty r_i\h^i\;,
\end{equation}
where the first term belongs to $\cA[[\h]]$, with all $\beta_k$ belonging to $\Span_\bbK\cB$, while the second term
belongs to $\cI_{q(\h)}$, since all $r_i$, which are formal power series in $\h$, belong to it (recall that the
ideal $\cI_{q(\h)}$ is closed). The assignment $A\mapsto\sum_{k=0}^\infty\beta_k\h^k$ defines a surjective
$\bbK[[\h]]$-linear map $\fA[[\h]]\to\cA[[\h]]$. Since the decomposition \eqref{eq:A_decomp} of $A$ is unique, its
kernel is $\cI_{q(\h)}$. It follows that $\fA[[\h]]/\cI_{q(\h)}$ and $\cA[[\h]]$ are isomorphic, as was to be
shown.
\end{proof}

In view of the proposition, using the monomial basis $\cB$ of $\cA_\h$ we may identify $\cA_\h$ with~$\cA[[\h]]$ as
$\bbK[[\h]]$-modules and the product in $\cA_\h$ yields a product $\star$ on $\cA[[\h]]$, making
$(\cA[[\h]],\star)$ into a (formal) deformation of $\cA$. Under this identification we may write any $A\in\cA_\h$
uniquely as
\begin{equation}\label{eq:norm}
  A=\mathbf{n}_0(A)+\h\mathbf{n}_1(A)+\h^2\mathbf{n}_2(A)+\cdots\;.
\end{equation}
Here, every $\mathbf{n}_i(A)$ belongs to $\Span_\bbK\cB$ and the maps~$\mathbf{n}_i$ are linear maps $\cA_\h\to\cA$.

\begin{remark}\label{rem:several_parameters}
The adaptation to the case of several parameters is straightforward: one puts $\q(\h)=\q_0+\h\q_1+\cO\(\h^2\)$, where
$\q_0$ is a regular value for $(q_1,\dots,q_\ell)$ and $\q_1$ is any non-zero vector. \emph{Regular} means as in
the single parameter case that $\q_0$ is a regular value for all the rational coefficients that appear when writing
$x_jx_i$ for $j>i$ in terms of the $\bbK(\q)$-basis $\cB$ of $\cA_\q$.  The algebra $\cA_\h$ is constructed as
before and from there on everything proceeds as in the one-parameter case.
\end{remark}

\begin{remark}\label{rem:rescaling}
In the examples we will always set $q(\h)=q_0+\h$. In fact, the higher order terms in $\h$ do not play a role in
the construction of the Poisson structure; also, in the one-parameter case the value of the coefficient of $\h$ is
not very important since it amounts to rescaling $\h$, so we will always pick it equal to $1$.
\end{remark}

\section{Examples and applications}\label{sec:ex1}

In this section we show on a few different families of examples how to determine explicitly the Poisson algebra,
the reduced Poisson algebra and the Poisson module, associated with a (formal) deformation; all our examples are
obtained from quantum algebras, as explained in the previous section. As a first application, we use the (reduced)
Poisson algebra to show that two particular deformations of some algebra which is not commutative are not
equivalent. As a second application, we will show in the next section how the Poisson module is used for obtaining
a Hamiltonian formulation of nonabelian systems.

\subsection{Computing the Poisson brackets}\label{par:compute}

We first explain a few general facts on the computation of the Poisson brackets $\PB$ and $\PB_1$ on $\Pi(\cA)$ and
on $Z(\cA)$ respectively, and on the Lie action $\PBm$ of $\Pi(\cA)$ on $\cA$. Recall that these are associated
with a deformation $(\cA[[\h]],\star)$ of some (not necessarily commutative) algebra $\cA$ and that the deformation
itself is associated with a quantum algebra $\cA_\q$ and a regular value $\q_0$ of $\q$; in the examples that
follow there is a single parameter, $\q=q$. Also, in these examples we always take as base ring $R$ the field of
complex numbers~$\bbC$ so that $q_0$ is just a complex number.

In each one of the examples, the Poisson brackets $\PB$ and $\PB_1$ will be computed explicitly for generators of
$\Pi(\cA)$ and of $Z(\cA)$ respectively; also, the Lie action $\PBm$ will be computed for generators of $\Pi(\cA)$
and of~$\cA$. In view of the Leibniz identity, this yields the Poisson bracket for any pair of elements of
$\Pi(\cA)$; see section \ref{par:limit} for the case of the Lie action. In the case of the Poisson brackets this
can be done by using a formula that generalizes the classical formula for computing Poisson brackets on $\bbC^M$ in
terms of the Poisson brackets between the coordinate functions on $\bbC^M$. If, say, $X_1,X_2,\dots,X_M$ are
generators of a commutative Poisson algebra and $P=P(X_1,X_2,\dots,X_M)$ and $Q=Q(X_1,X_2,\dots,X_M)$ are two
elements expressed in terms of these generators, then
\begin{equation}\label{eq:PQ_bracket}
  \pb{P,Q}=\sum_{i,j=1}^M\pp P{X_i}\pp Q{X_j}\pb{X_i,X_j}\;.
\end{equation}
For a more formal statement and a proof, see \cite[Prop.\ 1.9]{PLV}; notice that the cited proposition says in
particular that the right hand side of \eqref{eq:PQ_bracket} is independent of the way in which $P$ and $Q$ are
expressed in terms of the generators. The skew-symmetric matrix $\left(\pb{X_i,X_j}\right)_{1\leqslant i,j\leqslant
  M}$, which in view of \eqref{eq:PQ_bracket} completely determines the Poisson bracket, is called the
\emph{Poisson matrix of} $\PB$ \emph{with respect to the generators $X_1,\dots,X_M$}.

We will need to determine the center $Z(\cA)$ of $\cA$ and find generators for it. In general this can be very
complicated, but in most 
examples that follow we have for every $1\leqslant j<i\leqslant N$ a relation in $\cA$ of the form
$x_ix_j=q_{i,j}x_jx_i$, with $q_{i,j}\in\bbC^*$. As we already pointed out in Example \ref{exa:typical}, on the one
hand this implies that we have a monomial basis of $\cA$ whose elements are of the form $x_1^{i_1}x_2^{i_2}\dots
x_N^{i_N}$, and on the other hand that the commutator of two such monomials is a multiple (element of $\bbC$) of
their product. This in turn implies that when an element of the center is written in the monomial basis, each one
of its terms belongs to the center. The center of $\cA$ is therefore generated by monomials in $x_1,\dots,x_N$. The
same argument shows that for such commutation relations the center of $\cA_\q$ is generated by monomials (the
multiple is then an element of $\bbC(\q)$).

In order to write explicit formulas for the Poisson brackets and for the Lie action, we will need to find algebra
generators for $\Pi(\cA)=Z(\cA)\times (\cA/Z(\cA))$, where we recall that the associative product in $\Pi(\cA)$ is
denoted by $\cdot$ and is given by~\eqref{eq:prod}. Such generators can be chosen in the union of
$Z(\cA)\times\set{\cl0}$ and $\set0\times\cA/Z(\cA)$ since $\Pi(\cA)$ is the direct sum of these subspaces. For
such generators, \eqref{eq:prod} simplifies to
\begin{gather}
  (a,\cl0)\cdot(b,\cl0)=(ab,\cl{(a,b)_1})
  \;,\label{eq:prod_simp_1}\\
  (0,\cl{a_1})\cdot(0,\cl{b_1})=\(0,\cl0\)\;,\label{eq:prod_simp_2}\\
  (a,\cl0)\cdot(0,\cl{b_1})=\(0,\cl{ab_1}\)\;,\label{eq:prod_simp_3}
\end{gather}
where we recall that $(a,b)_1$ is the coefficient in $\h$ of $a\star b$ (see \eqref{eq:prod_exp}).  Notice that
when $(a,b)_1\in Z(\cA)$, \eqref{eq:prod_simp_1} simplifies further to
\begin{equation}\label{eq:prod_simp_4}
  (a,\cl0)\cdot(b,\cl0)=\(ab,\cl0\)\;.
\end{equation}
This happens for example when the commutator of $a$ and $b$ in $\cA_\h$ is a constant (element of $R[[\h]]$), or is
a multiple of their product $ab$; for the latter, see again Example \ref{exa:typical}. In general, a generating set
of $\Pi(\cA)$ can be constructed using the following proposition:

\begin{prop}\label{prop:generators}
  Suppose that $z_1,\dots,z_k$ are generators of $Z(\cA)$ as a (unital) algebra and that
  $\cl{t_1},\dots,\cl{t_\ell}$ are generators for $\cA/Z(\cA)$ as a $Z(\cA)$-module. Denote $Z_i:=(z_i,\cl0)$ and
  $T_j:=(0,\cl{t_j})$ for $i=1,\dots,k$ and $j=1,\dots,\ell$. Then $Z_1\dots,Z_k,T_1,\dots,T_\ell$ are algebra
  generators of $(\Pi(\cA),\cdot)$.
\end{prop}
\begin{proof}
  Let $(Z,\cl A)\in\Pi(\cA)$. Since $z_1,\dots,z_k$ are generators of $Z(\cA)$, there exists a polynomial $P$ such
  that $P(z_1,\dots,z_k)=Z$.  In view of \eqref{eq:prod_simp_1} -- \eqref{eq:prod_simp_3}, $P(Z_1,\dots,Z_k)=(Z,\cl
  T)$ where $T\in\cA$. Since $\cA/Z(\cA)$ is generated by $\cl{t_1},\dots,\cl{t_\ell}$ as a $Z(\cA)$-module, there
  exist $\alpha_1,\dots,\alpha_\ell\in Z(\cA)$ such that $\cl{A-T}=\sum_{j=1}^\ell\alpha_j\cl{t_j}$. Writing each
  of these $\alpha_j$ as $\alpha_j=P_j(z_1,\dots,z_k)$, where each $P_j$ is a polynomial, we get
  \begin{equation*}
    (Z,\cl A)=P(Z_1,\dots,Z_k)+(0,\cl{A-T})=P(Z_1,\dots,Z_k)+\sum_{j=1}^\ell P_j(Z_1,\dots,Z_k)\cdot
    T_j\;,
  \end{equation*}
  which expresses explicitly $(Z,\cl A)$ in terms of the elements $Z_1,\dots,Z_k$ and $T_1,\dots,T_\ell$, which
  shows that the latter are generators of $\Pi(\cA)$.
\end{proof}
Using the algebra generators $Z_1,\dots,Z_k$ and $T_1,\dots,T_\ell$ of $\Pi(\cA)$, we get a surjective algebra
homomorphism $\bbC[Z_1,\dots,Z_k,T_1,\dots,T_\ell]\to\Pi(\cA)$ and we can describe $(\Pi(\cA),\cdot)$ as the
quotient $\bbC[Z_1,\dots,Z_k,T_1,\dots,T_\ell]/\cJ$, where $\cJ$ is the kernel of the homomorphism; we will
explicitly compute this kernel in our examples, providing thereby the algebra structure of $\Pi(\cA)$.

The Poisson brackets of the chosen generators can then be explicitly computed from the following formulas, which
are a specialization of \eqref{eq:PB}:
\begin{gather}
  \pb{(a,\cl0),(b,\cl0)}=\(\pb{a,b}_1,\cl{\pb{a,b}_2}\) \;,\label{eq:bra_1}\\
  \pb{(0,\cl{a_1}),(0,\cl{b_1})}=\(0,\cl{\lb{a_1,b_1}}\)\;,\label{eq:bra_2}\\
  \pb{(a,\cl0),(0,\cl{b_1})}=\(0,\cl{\pb{a,b_1}_1}\)\;,\label{eq:bra_3}
\end{gather}
where $a,b\in Z(\cA)$ and $a_1,b_1\in\cA$, and so we only need to compute $\pb{a,b}_1$, $\pb{a,b}_2$,
$\lb{a_1,b_1}$ and $\pb{a,b_1}_1$ for such elements to determine these Poisson brackets. Recall from
\eqref{eq:comm_exp} that $\pb{a,b}_1$ and $\pb{a,b}_2$ are the coefficients in $\h$ and $\h^2$ of $a\star b$.  As
before, when $\pb{a,b}_2\in Z(\cA)$, \eqref{eq:bra_1} simplifies further to
\begin{equation}\label{eq:bra_4}
  \pb{(a,\cl0),(b,\cl0)}=\(\pb{a,b}_1,\cl0\)\;.
\end{equation}
Similarly the Lie action of $\Pi(\cA)$ on $\cA$ is given for $a\in Z(\cA)$ and $a_1,b\in\cA$ by the following
formulas, which are a specialization of \eqref{eq:final_module}:
\begin{equation}\label{eq:mod_simp}
\pbm{(a,\cl0)}b=\pb{a,b}_1\quad\hbox{ and }\quad \pbm{(0,\cl{a_1})}b=\lb{a_1,b}\;.
\end{equation}
In view of the following proposition, we will also be interested in the center $Z(\cA_\q)$ of the quantum algebra
$\cA_\q$.

\begin{prop}\label{prp:Cas_from_center}
  Suppose that $X=X_0+\h X_1+\cdots$ is a central element of $\cA[[\h]]$. Then $\(X_0,\cl{X_1}\)$ is a Casimir of
  $(\Pi(\cA),\PB)$ and belongs to the annihilator of $\cA$.
\end{prop}
\begin{proof}
  Let $\(Y_0,\cl{Y_1}\)$ be any element of $\Pi(\cA)$ and denote $Y=Y_0+\h Y_1\in\cA[[\h]]$. According to
  Proposition \ref{prp:pr_poisson},
  \begin{equation*}
    \pb{\(X_0,\cl{X_1}\),\(Y_0,\cl{Y_1}\)}=\pr{\lb{X_0+\h X_1,Y_0+\h Y_1}_\h}=\pr\lb{X,Y}_\h=0\;,
  \end{equation*}
  since $X$ belongs to the center of $\cA[[\h]]$. This shows that $\(X_0,\cl{X_1}\)$ Poisson commutes with all
  elements of $\Pi(\cA)$, i.e., is a \emph{Casimir} of $\Pi(\cA)$. Similarly, for $a\in\cA$ one has
  $\pbm{\(X_0,\cl{X_1}\)}a=0$ since $\lb{X,a}_\star=0$ for any $a\in\cA\subset\cA[[\h]]$.
\end{proof}
Let $A=A(\q)\in Z(\cA_\q)$, which we may assume to depend polynomially on $\q$. For any regular value $\q_0$ of
$\q$, and any non-zero $\q_1$, expand $A(\q(\h))=A(\q_0+\h\q_1)=X_0+\nu X_1+\cO(\h)$, which is a central element of
$\cA[[\h]]$, hence leads in view of the proposition to the Casimir $\(X_0,\cl{X_1}\)$ of $\Pi(\cA)$. In general,
not all Casimirs of $\Pi(\cA)$ are obtained in this way.

\subsection{The quantum plane}\label{par:quantum_plane}
As a first example, we consider the (complex) quantum plane, which is defined as being the noncommutative algebra
\begin{equation}\label{eq:quantum_plane}
  \cA_q:=\bbC_q[x,y]=\frac{\bbC(q)\langle x,y\rangle}{\langle yx-qxy\rangle}\;.
\end{equation}
As a basis $\cB$ for this quantum algebra we take the normally ordered monomials $x^my^n,$ $m,n\in\bbN$. It is a
PBW basis and any $q_0\in\bbC$ is a regular value of it. In $\cA_q$ we have $yx=qxy$, or equivalently
$\lb{x,y}=(1-q)xy$. As we already pointed out in Section \ref{par:compute}, this implies that the center of $\cA_q$
is generated by monomials, from which it is clear that the center of $\cA_q$ consists of constants only,
$Z(\cA_q)=\bbC(q)$.

The evaluation of $\cA_q$ at $q=1$ is the polynomial algebra $\bbC[x,y]$, which is commutative, so let us consider
its evaluation at $q=-1$ to illustrate how to obtain the Poisson and reduced Poisson algebras associated with the
deformation; see below for other values of $q$. As explained in Section~\ref{par:quotients_and_deformations}, we
set $q(\h)=-1+\h$ and consider
\begin{equation*}
  \cA_\h=\frac{\bbC[[\h]]\langle x,y\rangle}{\langle yx-(\h-1)xy\rangle}\simeq\cA[[\h]]\;,\hbox{ where }
  \cA:=\frac{\bbC\langle x,y\rangle}{\langle yx+xy\rangle}\;.
\end{equation*}
The product $\star$ on $\cA[[\h]]$ is induced by the above isomorphism and is completely specified by $y\star
x=(\h-1)x\star y$. In $\cA$ we have $yx=-xy$, or equivalently $\lb{x,y}=2xy$, so that $Z(\cA)$ is also generated by
monomials.  Since $x$ and $y$ anticommute in $\cA$, the center $Z(\cA)$ consists of all polynomials that are even
in $x$ and in~$y$, hence is generated by $x^2$ and $y^2$. Then $\cA/Z(\cA)$ is generated as a $Z(\cA)$-module by
$\cl x$, $\cl y$ and $\cl{xy}$. According to Proposition \ref{prop:generators}, the following~$5$ elements are
algebra generators of~$\Pi(\cA)$:
\begin{equation*}
  X=\(x^2,\cl0\)\;,\  Y=\(y^2,\cl0\)\;,\ U=\(0,\cl x\)\;,\ V=\(0,\cl y\)\;,\ W=\(0,\cl{xy}\)\;,
\end{equation*}
with product $\cdot$ given by \eqref{eq:prod_simp_2} -- \eqref{eq:prod_simp_4}.  The natural algebra homomorphism
$\bbC[X,Y,U,V,W]\to \Pi(\cA)$ has as kernel the ideal $\mathcal J$ of~$\Pi(\cA)$ generated by $U^2,V^2,W^2,U\cdot
V,V\cdot W,U\cdot W$. To show this, first notice that the kernel clearly contains these elements, while the
elements $X^i\cdot Y^j\cdot U^{\epsilon_1}V^{\epsilon_2}\cdot W^{\epsilon_3}$ with at most one $\epsilon_i$ equal
to 1 and the others equal to 0, are all linearly independent. Indeed, they are given by
\begin{equation*}
  X^i\cdot Y^j=\(x^{2i}y^{2j},\cl0\)\;,\quad X^i\cdot Y^j\cdot U^{\epsilon_1}\cdot V^{\epsilon_2}\cdot
  W^{\epsilon_3}=\(0,\cl{x^{2i+\epsilon_1+\epsilon_3}y^{2j+\epsilon_2+\epsilon_3}}\)\;,
\end{equation*}
where exactly one of the $\epsilon_i$ is equal to $1$.  Therefore, as an algebra,
$(\Pi(\cA),\cdot)\simeq\bbC[X,Y,U,V,W]/{\mathcal J}$, while $Z(\cA)\simeq\bbC[X,Y]$.
The Poisson brackets $\PB$ between the generators $X,Y,\dots,W$ are given
in Table \ref{tab:vol_1}.
\begin{table}[h]
  \def\arraystretch{1.8}
  \setlength\tabcolsep{0.4cm}
  \centering
\begin{tabular}{c|ccccc}
   $\PB$&$X$&$Y$&$U$&$V$&$W$\\
  \hline
  $X$&$0$&$4X\cdot Y$&$0$&$2X\cdot V$&$2X\cdot W$\\
  $Y$&$-4X\cdot Y$&$0$&$-2U\cdot Y$&$0$&$-2Y\cdot W$\\
  $U$&$0$&$2U\cdot Y$&0&$2W$&$2X\cdot V$\\
  $V$&$-2X\cdot V$&$0$&$-2W$&$0$&$-2Y\cdot U$\\
  $W$&$-2X\cdot W$&$2Y\cdot W$&$-2X\cdot V$&$2Y\cdot U$&0\\
\end{tabular}
\bigskip
\caption{}\label{tab:vol_1}
\end{table}
They were computed using \eqref{eq:bra_1} -- \eqref{eq:bra_3}. Let us show for example that
$\pb{X,Y}=\(4x^2y^2,\cl0\)=4X\cdot Y$. For the first equality, one needs to verify in view of \eqref{eq:bra_4} that
$\pb{x^2,y^2}_1=4x^2y^2$. To do this, we first compute, in $\cA_\h$,
\begin{equation*}
  \lb{x^2,y^2}=\(1-q(\h)^4\)x^2y^2=\(1-(\h-1)^4\)x^2y^2=4\h x^2y^2+\cO\(\h^2\)\;,
\end{equation*}
from which it follows that $\lb{x^2,y^2}_\star=4\h x^2y^2+\cO\(\h^2\),$ as we needed to show.  For the second
equality, it suffices to notice that the product $X\cdot Y$ is given by \eqref{eq:prod_simp_4} because
$\(x^2,y^2\)_1=0$.  Similarly, $\pb{U,V}=2W$ since $[x,y]=2xy$. As a last example,
$\pb{X,W}=\(0,2\cl{x^3y}\)=X\cdot W$ because $\pb{x^2,xy}_1=2x^3y$, which follows from
\begin{equation*}
  \lb{x^2,xy}=\(1-q(\h)^2\)x^3y=\(1-(\h-1)^2\)x^3 y=2\h x^3y+\cO\(\h^2\)\;.
\end{equation*}
It is clear from the table that the reduced Poisson algebra $(Z(\cA),\PB_1)$ can be described as the polynomial
algebra $\bbC[X,Y]$, with Poisson bracket $\pb{X,Y}_1=4XY$. The action and Lie action of the generators of
$\Pi(\cA)$ on the generators $x$ and $y$ of $\cA$ is given in Table \ref{tab:vol_2}.
\begin{table}[h]
  \def\arraystretch{1.8}
  \setlength\tabcolsep{0.4cm}
\begin{tabular}{c|cc}
   $\cdot$&$x$&$y$\\
  \hline
  $X$&$x^3$&$x^2y$\\
  $Y$&$xy^2$&$y^3$\\
    $U$&$0$&$0$\\
  $V$&$0$&$0$\\
$W$&$0$&$0$\\
\end{tabular}
\qquad\qquad
\begin{tabular}{c|cc}
   $\PBm$&$x$&$y$\\
  \hline
  $X$&$0$&$2x^2y$\\
  $Y$&$-2xy^2$&$0$\\
    $U$&$0$&$2xy$\\
  $V$&$-2xy$&$0$\\
$W$&$-2x^2y$&$2xy^2$\\
\end{tabular}
\bigskip
\caption{}\label{tab:vol_2}
\end{table}
The entries of the rightmost table are computed from \eqref{eq:mod_simp}. For example, $\pbm
Xy=\pbm{\(x^2,\cl0\)}y=\pb{x^2,y}_1=2x^2y$, where the last equality is obtained from
\begin{equation*}
  \lb{x^2,y}=(1-q(\h)^2)x^2 y=\(1-(1-\h)^2\)x^2y=2\h x^2y+\cO(\h^2)\;.
\end{equation*}
Also, $\pbm Uy=\pbm{\(0,\cl x\)}y=\lb{x,y}=2xy$, since in $\cA$, $x$ and $y$ anticommute.

We now consider arbitrary values of $q$. Since $\lb{x^iy^j,x}=(q^j-1)x^{i+1}y^j$ in $\cA_q$, the center of $\cA$
will be $\bbC$, unless $q$ is a root of unity. Let us therefore consider a primitive $n$-th root of unity $\xi$
(where $n>1$). Setting $q(\h)=\xi+\h$, the above considerations and computations for $n=2$ are easily
generalized. First, since we know that $Z(\cA)$ is generated by monomials, it is easily checked as above that $x^n$
and $y^n$ generate $Z(\cA)$ and from it that $X:=\(x^n,\cl0\)$, $Y:=\(y^n,\cl0\)$ and $W_{i,j}:=\(0,\cl{x^iy^j}\)$,
where $0\leqslant i,j< n,\ i+j\neq0$ generate~$\Pi(\cA)$. As an algebra,
$(\Pi(\cA),\cdot)\simeq\bbC[X,Y,W_{i,j}]/\langle W_{i,j}W_{k,\ell}\rangle$, where the indices $i,j,k,\ell$ are in
the range $0,\dots,n-1$, with $i+j\neq0$ and $k+\ell\neq0$. The Poisson brackets of these generators are computed
as above and are given in Table \ref{tab:intro_1} in the introduction.

Notice that $W_{i+k,j+\ell}$ is for large values of $i,j,k,\ell$ not one of the chosen generators of $\Pi(\cA)$ but
can easily be rewritten in terms of these generators using $W_{\alpha n+\beta,\alpha'n+\beta'}=X^\alpha \cdot
Y^{\alpha'}\cdot W_{\beta,\beta'}$.
The reduced Poisson algebra $(Z(\cA),\PB_1)$ can be described as the polynomial algebra $\bbC[X,Y]$, with Poisson
bracket $\pb{X,Y}_1=-\xi^{-1}nXY$.
The action and Lie action of $\Pi(\cA)$ on $\cA$ are given by Table~\ref{tab:intro_2}.
%

\begin{remark}\label{rem:quantum_torus}
The quantum plane \eqref{eq:quantum_plane} is closely related to the (two-dimensi\-onal) quantum torus, which is
defined as the noncommutative algebra
\begin{equation}\label{eq:quantum_torus}
  \cA_q:=\bbT_q[x,y]=\frac{\bbC(q)\langle x,y,x^{-1},y^{-1}\rangle}{\langle yx-qxy\rangle}\;.
\end{equation}
The above considerations and computations are easily adapted to this case. Still considering the case of $q$ being
an $n$-th root of unity, $Z(\cA)$ is now generated by $x^n,\ x^{-n},\ y^{n}$ and $y^{-n}$
and~$Z(\cA)\simeq\bbC[X,Y,X^{-1},Y^{-1}]$, where $X^{-1}:=\(x^{-n},\cl0\)$ and $Y^{-1}:=\(y^{-n},\cl0\)$ are two
extra generators for $\Pi(\cA)$. The ideal $\cJ$ has two extra generators $XX^{-1}-1$ and $YY^{-1}-1$. The above
tables containing the Poisson brackets of $\Pi(\cA)$ and the actions of $\Pi(\cA)$ on $\cA$ in the case of the
quantum plane still contain all information for the case of the quantum torus because by the Leibniz identity,
$\pB{X^{-1}}=-X^{-2}\pB{X}$ and $\pBm{X^{-1}}=-X^{-2}\cdot\pBm{X}$, and similarly for the brackets with
$Y^{-1}$. In more formal terms, for any value of $q\in\bbC^*$, the Poisson algebra $\Pi(\cA)$ for the quantum torus
is the localization of the Poisson algebra $\Pi(\cA)$ of the quantum plane, with respect to the multiplicative
system of $\cA$ generated by $X$ and~$Y$ (see~\cite[Section~2.4.2]{PLV}). The algebra $\cA$, which is a Poisson
module over $\Pi(\cA)$, then becomes a Poisson module over this localization of $\Pi(\cA)$.
\end{remark}

\subsection{A quantum algebra related to the Volterra chain}\label{par:volterra_1}
We next consider a more elaborate example which is related to the $N$-periodic nonabelian Volterra chain (see
Section \ref{par:volterra}). The quantum algebra, which is a particular case of Example \ref{exa:typical}, is
generated by $x_1,\dots,x_N$, which all commute except the neighboring pairs $x_{i+1}x_i=qx_ix_{i+1}$ for $i=1,\dots,N$; in this
formula the index $i$ of $x_i$ is taken modulo $N$, so that $x_1x_N=qx_Nx_1$. For $1\leqslant i,j\leqslant N$ let
us denote their $N$-periodic distance by $d_N(i,j)$; so $d_N(i,j)=\min\set{\vert i-j\vert,N-\vert i-j\vert}$.  Then
$x_ix_j=x_jx_i$ when $d_N(i,j)\neq1$. We therefore consider
\begin{equation}\label{eq:quantum_volterra}
  \cA_q:=\frac{\bbC(q)\langle x_1,\dots,x_N\rangle}{\cI_q}\;,\hbox{ where } \cI_q={\langle
    x_{i+1}x_i-qx_ix_{i+1},x_ix_j-x_jx_i\rangle}_{d_N(i,j)\neq1}\;.
\end{equation}
An automorphism of order $N$ of $\cA_q$ is defined by $\cS(x_i):=x_{i+1}$ for all $i$.  We use as a basis $\cB$ of
$\cA_q$ the normally ordered monomials $x_1^{i_1}\dots x_N^{i_N}$ with $i_1,\dots,i_N\in\bbN$.  Notice that $q_0=0$
is not a regular value of $\cB$, since $x_Nx_1=q^{-1}x_1x_N$, but all other values are regular values.  Again,
$\cB$ is a PBW basis. As we explained in Section \ref{par:compute}, the center of~$\cA_q$ is in this case generated
by monomials. We use this fact to show that $Z(\cA_q)$ is generated by $x_1x_2\dots x_N$ when~$N$ is odd, while it
is generated by $x_1x_3\dots x_{N-1}$ and $x_2x_4\dots x_{N}$ when $N$ is even. A monomial $x_1^{i_1}x_2^{i_2}\dots
x_N^{i_N}$ of~$\cA_q$ belongs to the center if and only if its commutator (in $\cA_q$) with any~$x_\ell$
vanishes. From
\begin{equation}\label{eq:com_in_volterra}
  0=\lb{x_1^{i_1}\dots x_N^{i_N},x_\ell}=(q^{i_{\ell+1}-i_{\ell-1}}-1)x_\ell x_1^{i_1}\dots x_N^{i_N}
\end{equation}
it follows that if $x_1^{i_1}\dots x_N^{i_N}$ belongs to the center of $\cA_q$, then $i_k=i_{k+2}$ for
$k=1,\dots,N$, which yields the claim (recall that the indices of $x$ are $N$-periodic, so that also $i_{N+1}=i_1$
and $i_0=i_N$).  It is also clear from \eqref{eq:com_in_volterra} that for values of $q$ that are not roots of
unity, the center of the corresponding algebra $\cA$ will contain no other elements and hence
$Z(\cA)\times\set{\cl0}\subset \Pi(\cA)$ consists of Casimirs only (Proposition \ref{prp:Cas_from_center}) and the
only non-trivial Poisson brackets in $\Pi(\cA)$ are given by commutators (see \eqref{eq:bra_1} --
\eqref{eq:bra_3}). We will therefore consider the evaluation of $q$ at roots of unity only. Also, in view of the
difference between $N$ even and odd, we will first consider in detail the case that $N$ is odd and spell out
afterwards how to adapt the results in case~$N$ is even.

Let $N>2$ be odd and let $\xi$ denote a primitive $n$-th root of unity, $\xi^n=1$, where $n>1$; see Remark
\ref{rem:n=1} below for the case of $n=1$.  We set $q(\h)=\xi+\h$ and consider
\begin{equation*}
  \cA_\h=\frac{\bbC[[\h]]\langle x_1,\dots,x_N\rangle}{\cI_{q(\h)}}\simeq\cA[[\h]]\;,\hbox{ where }
  \cA:=\frac{\bbC\langle x_1,\dots,x_N\rangle}{\cI_\xi}\;.
\end{equation*}
It is clear that the center of $\cA$ contains $x_1^n,\dots,x_N^n$ and $x_1x_2\dots x_N$. We claim that $Z(\cA)$ is
generated by these elements. To show this, we look for monomials $x_1^{i_1}x_2^{i_2}\dots x_N^{i_N}$ in the center
of $\cA$, with $0\leqslant i_1,i_2,\dots, i_N<n$; as above, it follows however from \eqref{eq:com_in_volterra} with
$q=\xi$ that then all $i_k$ must be equal (recall that $N$ is assumed odd), and so the monomial is of the form
$x_1^k x_2^k\dots x_N^k$, for some $k$. This shows our claim. It is then also clear that $\cA$ is generated as a
$Z(\cA)$-module by the monomials $x_1^{i_1}x_2^{i_2}\dots x_N^{i_N},$ where $0\leqslant i_1,\dots,i_N<n$ are not
all zero and at least one of them is zero. In view of Proposition \ref{prop:generators}, $\Pi(\cA)$ is generated by
\begin{equation*}
  X_1:=(x_1^n,\cl0)\;,\dots, X_N:=(x_N^n,\cl0)\;,\ X:=(x_1x_2\dots x_N,\cl0) \hbox{ and }
  W_{i_1,\dots,i_N}:=\(0,\cl{x_1^{i_1}x_2^{i_2}\dots x_N^{i_N}}\)\;,
\end{equation*}
where $0\leqslant i_1,\dots,i_N<n$ are not all zero and at least one of them is zero. Since $\Pi(\cA)$ is generated
by these elements, we have a surjective morphism $\bbC[X_i,X,W_{i_1,\dots,i_N}]\to\Pi(\cA)$, with kernel
\begin{equation}\label{eq:cI_n}
  \cJ=\langle{X_1\cdot X_2\cdots X_N-X^n,W_{i_1,\dots,i_N}\cdot X-W_{i_1+1,\dots,i_N+1},W_{i_1,\dots,i_N}\cdot
    W_{j_1,\dots,j_N}}\rangle\;,
\end{equation}
where the indices $i_1,\dots,i_N$ and $j_1,\dots,j_N$ are as above. It is understood that when one of the indices
of the term $W_{i_1+1,\dots,i_N+1}$ in $\eqref{eq:cI_n}$ is at least $n$, it is rewritten in terms of the
generators, for example $W_{n,i_2,\dots,i_N}=X_1W_{0,i_2,\dots,i_N}$ and $W_{n,0,0,\dots,0}=0$. It follows that, as
algebras, $\Pi(\cA)\simeq\bbC[X_i,X,W_{i_1,\dots,i_N}]/\cJ$ and $Z(\cA)\simeq\bbC[X_i,X]/\langle X_1X_2\dots
X_N-X^n\rangle$.

In terms of these generators, the Poisson bracket $\PB$ of $\Pi(\cA)$ is given by Table \ref{tab:vol_5}, in which
$I=(i_1,\dots,i_N)$ and $J=(j_1,\dots,j_N)$.
%
%
%
\begin{table}[h]
  \def\arraystretch{1.8}
  \setlength\tabcolsep{0.4cm}
  \centering
\begin{tabular}{c|ccc}
   $\PB$&$X_\ell$&$W_J$\\
  \hline
  $X_k$&$n^2(\delta_{k,\ell+1}-\delta_{k,\ell-1})\xi^{-1}X_k\cdot X_\ell$&$({j_{k-1}}-{j_{k+1}})\xi^{-1}nX_k\cdot W_J$\\
  $W_I$&$({i_{\ell+1}}-{i_{\ell-1}})\xi^{-1}nX_\ell\cdot W_I$& $\(\xi^{\eta(I,J)}-\xi^{\eta(J,I)}\)W_{I+J}$\\
\end{tabular}
\bigskip
\caption{}\label{tab:vol_5}
\end{table}
In this table, the indices of $i$ and $j$ are again taken modulo~$N$, so that $i_{N+1}=i_1$ and $i_0=i_N$, and
similarly for $j$. Notice that in view of the automorphism $\cS$ we may assume that $k\neq1,N$, or that
$\ell\neq1,N$, when verifying the entries of the table. Also, $\eta(I,J)$ is a shorthand for
$\sum_{s=1}^{N-1}i_{s+1}j_s-i_Nj_1$ and is computed from
\begin{equation*}
  (x_1^{i_1}\dots x_{N}^{i_{N}})(x_1^{j_1}\dots x_{N}^{j_{N}})=\xi^{\eta(I,J)}x_1^{i_1+j_1}\dots
  x_{N}^{i_{N}+j_{N}}\;,
\end{equation*}
using the relations in $\cA$ which say that all variables $x_i$ commute, except the neighboring pairs $x_{i+1}x_i=\xi x_ix_{i+1}$ for all $i\,\mbox{mod}\,N$.
As before, when $W_{I+J}$ is not one of the chosen generators of $\Pi(\cA)$ it is easily written as a product of
such generators. We did not put the generator $X$ in the table because it is a Casimir, so the brackets of $X$ with
any other generator (or element of $\Pi(\cA)$) is zero; it is a Casimir because $x_1x_2\dots x_N$ is a central
element of $\cA_q$, see Proposition~\ref{prp:Cas_from_center}. The reduced Poisson algebra $(Z(\cA),\PB_1)$ is
easily read off from the upper-left corner of Table \ref{tab:vol_5}. It is a so-called \emph{diagonal} Poisson
structure (see \hbox{\cite[Example 8.14]{PLV}}).

For future use, we also compute the Poisson module structure on $\cA$, still assuming $N$ odd. It is clear that the
action of $\Pi(\cA)$ on $\cA$ is given by $X_k\cdot x_\ell=x^n_kx_\ell$, $W_I\cdot x_\ell=0$ for all
$I,k,\ell$. According to \eqref{eq:final_module} and since $X$ is in the annihilator of $\cA$ (see
Proposition~\ref{prp:Cas_from_center}), $\pbm X{x_\ell}=0$. With $k,\ell=1,\dots,N$ and $I=(i_1,\dots,i_N)$
($0\leqslant i_1,\dots,i_N<n$),
\begin{equation}\label{eq:mod_volt}
  \pbm{X_k}{x_\ell}\stackrel{\eqref{eq:final_module}}=\pb{x_k^n,x_\ell}_1=(\delta_{k,\ell+1}-\delta_{k,\ell-1})
  \xi^{-1}nx_k^nx_\ell\;,
\end{equation}
and
\begin{equation}\label{eq:mod_volt_2}
  \pbm{W_I}{x_\ell}=\lb{x_1^{i_1}x_2^{i_2}\dots x_N^{i_N},x_\ell}=\alpha(I,\ell)x_1^{i_1}\dots
  x_\ell^{i_\ell+1}\dots x_N^{i_N}\;,
\end{equation}
where
\begin{equation*}
  \alpha(I,\ell)=\alpha(i_1,\dots,i_N,\ell)=\left\{
  \begin{array}{cl}
    \xi^{i_2-i_{N}}-1&\ell=1\;,\\
    \xi^{i_{\ell+1}}-\xi^{i_{\ell-1}}&1<\ell<N\;,\\
    1-\xi^{i_{N-1}-i_1}&\ell=N\;.
  \end{array}
  \right.
\end{equation*}
These results were obtained from \eqref{eq:mod_simp} using the commutation rules in $\cA_\h\simeq\cA[[\h]]$. For
example, $\lb{x_k^n,x_\ell}_\star=0$ when $d_N(k,\ell)\neq1$, since then $x_k$ and $x_\ell$ commute. Also, if
$k\neq N$ then $\lb{x_k^n,x_{k+1}}=(1-q^n)x_k^nx_{k+1}=(1-(\xi+\h)^n)x_k^nx_{k+1}= -\xi^{-1}n\h
x_k^nx_{k+1}+\cO\(\h^2\)$, so that $\lb{x_k^n,x_{k+1}}_\star= -\xi^{-1}n\h x_k^nx_{k+1}+\cO\(\h^2\)$ and hence
$\pbm{X_k}{x_{k+1}}=\pb{x_k^n,x_{k+1}}_1=-\xi^{-1}n\h x_k^nx_{k+1}$; this result is also valid for $k=N$ since
$x_N^nx_1=x_1x_N^n$ hence belongs to $\cB$. The formulas in \eqref{eq:mod_volt_2} are just commutators in $\cA$, so
they follow immediately from the relations in $\cA$, which are given by $x_{i+1}x_i=\xi x_ix_{i+1}$, and
$x_ix_j=x_jx_i$ when $d_N(i,j)\neq1$.

We now consider the minor adaptations to be done in case $N$ is even, $N>2$. As we have already shown, $Z(\cA_q)$
is generated by $x_1x_3\dots x_{N-1}$ and $x_2x_4\dots x_{N}$. By the same proof as above, $Z(\cA)$ is generated by
\begin{equation*}
  X_i:=(x_i^n,\cl0)\;,\ Y_1:=(x_1x_3\dots x_{N-1},\cl0)\;,\ Y_2:=(x_2x_4\dots x_{N},\cl0)\;,
\end{equation*}
where $i=1,\dots,N$, and
\begin{equation*}
  W_{i_1,\dots,i_N}:=\(0,\cl{x_1^{i_1}x_2^{i_2}\dots x_N^{i_N}}\)\;,\qquad 0\leqslant i_1,\dots,i_N<n\;,
\end{equation*}
where not all indices $i_k$ are zero, where at least one index $i_k$ with $k$ even is zero, as well as at least one
index $i_k$ with $k$ odd. As an algebra, $\Pi(\cA)\simeq\bbC[X_i,Y_1,Y_2,W_{i_1,\dots,i_N}]/\cJ'$, where
\begin{align*}
  \cJ'=\langle X_1\cdot X_3\cdots X_{N-1}-Y_1^n,X_2\cdot X_4\cdots X_{N}-Y_2^n,W_{i_1,\dots,i_N}\cdot Y_1-
  W_{i_1+1,i_2,\dots,i_N},\\
  W_{i_1,\dots,i_N}\cdot Y_2-W_{i_1,i_2+1,\dots,i_{N+1}},W_{i_1,\dots,i_N}\cdot W_{j_1,\dots,j_N}\rangle\;,\qquad\qquad
\end{align*}
and $Z(\cA)\simeq\bbC[X_i,Y_1,Y_2]/\langle X_1X_3\dots X_{N-1}-Y_1^n,X_2X_4\dots X_{N}-Y_2^n\rangle$.  The table of
Poisson brackets for the generators is again given by Table \ref{tab:vol_5}, where we now leave out the zero rows
and columns corresponding to the Casimirs $Y_1$ and $Y_2$. Also, the formulas for the actions of $\Pi(\cA)$ on
$\cA$ are the same. So finally the cases $N$ even and odd are formally not very different.

\begin{remark}\label{rem:n=1}
When $n=1$, $\xi=1$ and $\cA$ is commutative, $\cA=Z(\cA)$, so that we are in the classical case. Then $\cA/Z(\cA)$
is trivial and $\cA$ is generated by $X_1,\dots,X_N$, which now take the simple form $X_i=(x_i,\cl0)$,
i.e.\ $X_i=x_i$ under the natural identification of $\Pi(\cA)$ with $\cA$. The above computation of the Poisson
brackets is still valid and leads to the following non-zero Poisson brackets between the $x_i$:
\begin{equation}\label{eq:volt_class}
  \pb{x_i,x_j}=(\delta_{i,j+1}-\delta_{i,j-1})x_ix_j\;.
\end{equation}
It is the standard Poisson structure of the (periodic) Volterra chain.
\end{remark}

\begin{remark}\label{rem:non_periodic}
The infinite case ($N=\infty$) is not very different and is in a sense simpler. The formula for $\eta(I,J)$ in
Table \ref{tab:vol_5} simplifies to $\eta(I,J)=\sum i_{s+1}j_s$ and \eqref{eq:mod_volt_2} becomes
\begin{equation*}
  \pbm{W_I}{x_\ell}=(\xi^{i_{\ell+1}}-\xi^{i_{\ell-1}}) x_1^{i_1}\dots x_\ell^{i_\ell+1}\dots x_N^{i_N}\;.
\end{equation*}
The Poisson algebra $\Pi(\cA)$ does not have any Casimirs; it can be extended to a larger Poisson algebra for which
the infinite products $\prod\limits_{i\textrm{ odd}}x_i$ and $\prod\limits_{i\textrm{ even}}x_i$ are Casimirs.
\end{remark}


\subsection{Another quantum algebra related to the Volterra chain}\label{eq:volterra_second}
We now describe the Poisson algebra corresponding to another quantum algebra, also related to the periodic Volterra
chain. The algebra is given by
\begin{equation*}
  \cA_q:=\frac{\bbC(q)\langle x_1,\dots,x_{N}\rangle} {\langle
    x_{i+1}x_i-(-1)^iqx_ix_{i+1},x_ix_j+x_jx_i\rangle_{d_N(i,j)\neq1}}\;,
\end{equation*}
with $N$ even, $N=2M$.  As before, the indices are periodic modulo $N$ and $d_N(i,j)$ denotes the periodic distance
between $i$ and $j$. As a basis of $\cA_q$ we take the monomials $x_1^{i_1}\dots x_{N}^{i_{N}}$, with
$i_1,\dots,i_{N}\in\bbN$. It is a PBW basis and again only $q_0=0$ is not a regular value of $\cB$. Notice that
since $N$ is even we can speak unambiguously of even and odd indices. A $\bbC$-algebra automorphism $\cS_{q}$ of
$\cA_q$ is defined by $\cS_{q}(x_i)=x_{i+1}$ and $\cS_{q}(q)=-q$. In view of the commutation relations in $\cA_q$, the
center of $\cA_q$ is again generated by monomials. Since
\begin{equation}\label{eq:comm_in_second}
  \lb{x_1^{i_1}\dots x_{N}^{i_{N}},x_\ell}=\((-1)^{\sum_{k=1}^Ni_k-i_\ell-i_{\ell+(-1)^\ell}}
  q^{i_{\ell+1}-i_{\ell-1}}-1\)x_\ell x_1^{i_1}\dots x_{N}^{i_{N}}\;,
\end{equation}
$x_1^{i_1}\dots x_{N}^{i_{N}}$ belongs to $Z(\cA_q)$ if and only if (1) all exponents that correspond to even
indices are equal, (2) all exponents that correspond to odd indices are equal and (3) the sum
${\sum_{k=1}^Ni_k-i_\ell-i_{\ell+(-1)^\ell}}$ is even for all $\ell$. Since this sum is of the form
${\sum_{k=1}^Ni_k-i_\ell-i_{\ell\pm1}}$, it contains $M-1$ terms $i_k$ with $k$ even and $M-1$ terms $i_k$ with $k$
odd; therefore, when $M$ is odd, condition (3) is satisfied as soon as conditions (1) and (2) are satisfied, while
when $M$ is even, (3) imposes furthermore that all the exponents have the same parity. It follows that
\begin{equation*}
  \hbox{$Z(\cA_q)$ is generated by }
    \left\{
    \begin{array}{ll}
      \prod\limits_{i\textrm{ odd}}x_i^2\;,\ \prod\limits_{i\textrm{ even}}x_i^2\;,\ \prod\limits_{i}x_i\;,\qquad
      &M\textrm{ even},\\[10pt]
      \prod\limits_{i\textrm{ odd}}x_i\;,\ \prod\limits_{i\textrm{ even}}x_i\;,\qquad&M\hbox{ odd}.
    \end{array}
    \right.
\end{equation*}

We will consider the evaluation of $\cA_q$ at $q=1$ only; notice that in view of the automorphism $\cS_{q}$ the
Poisson algebras obtained for $q=1$ and $q=-1$ are isomorphic. When $q=1$ the commutation relations in $\cA$ take
the simple form
\begin{equation}\label{eq:comm_A}
  x_{i+1}x_i=(-1)^ix_ix_{i+1}\;,\quad x_jx_i=-x_ix_j\;,\ d_N(i,j)\neq1\;,
\end{equation}
so all variables anticommute, except for half of the neighbors that commute. From these relations and from the fact
that $Z(\cA)$ is generated by monomials, it follows that
\begin{equation*}
  \hbox{$Z(\cA)$ is generated by }
    \left\{
    \begin{array}{ll}
      x_k^2\;,x_{2j}x_{2j+1}\;,\qquad &M\textrm{ even},\\[10pt]
      x_k^2\;,x_{2j}x_{2j+1}\;,\ \prod\limits_{i\textrm{ odd}}x_i\;,\
      \prod\limits_{i\textrm{ even}}x_i\;,\qquad&M\hbox{ odd},
    \end{array}
    \right.
\end{equation*}
where $k=1,\dots,N$ and $j=1,\dots,M$.

From these generators we construct, using Proposition \ref{prop:generators}, generators for~$\Pi(\cA)$. Let us
denote $X_k:=\(x_k^2,\cl0\)$ for $k=1,\dots,N$, $Y_j:=\(x_{2j}x_{2j+1},\cl0\)$ for $j=1,\dots,M$,
$Y:=\(\prod_{i\textrm{ odd}} x_i,\cl0\)$, $Z:=\(\prod_{i\textrm{ even}} x_i,\cl0\)$ and
$W_{i_1,\dots,i_{N}}:=\(0,\cl{x_1^{i_1}\dots x_{N}^{i_{N}}}\)$. Then $\Pi(\cA)$ is generated by
$X_1,\dots,X_{N},Y_1,\dots,Y_M$ and
\begin{equation*}
  \left\{
  \begin{array}{ll}
    \; W_{i_1,\dots,i_{N}}, \textrm{ $i_k\in\set{0,1}$, not all $i_k=0$, $i_{2j}i_{2j+1}=0\;,$}
    \qquad &M\textrm{ even},\\ [5pt]
    Y,\; Z,\; W_{i_1,\dots,i_{N}}, \textrm{ $i_k\in\set{0,1}$, not all $i_k=0$,  $i_{2j}i_{2j+1}=0$, not all
      $\left[\begin{array}{c} \textrm{even}\\ \textrm{odd}\end{array}\right] i_k=1\;,$}
        \qquad &M\textrm{ odd}.
  \end{array}
  \right.
\end{equation*}
When $M$ is even, $(\Pi(\cA),\cdot)\simeq\bbC[X_i,Y_j,W_{i_1,\dots,i_{N}}]/\cJ$, where
\begin{align*}
  \cJ:=\langle
  &Y_j^2-X_{2j}\cdot X_{2j+1},Y_j\cdot W_{\dots,i_{2j-1},1,0,i_{2j+2},\dots}-X_{2j}\cdot W_{\dots,i_{2j-1},0,1,i_{2j+2},\dots},\\
  &\quad Y_j\cdot W_{\dots,i_{2j-1},0,1,i_{2j+2},\dots}-X_{2j+1}\cdot W_{\dots,i_{2j-1},1,0,i_{2j+2},\dots}
  ,W_{i_1,\dots,i_{N}}\cdot W_{j_1,\dots,j_{N}}\rangle\;,
\end{align*}
and similarly when $M$ is odd.  Setting $q(\h)=1+\h$ we can now determine the Poisson brackets between the above
generators. Since $Y$ and $Z$ are constructed from central elements of $\cA_q$, they are Casimirs of
$(\Pi(\cA),\PB)$ (Proposition \ref{prp:Cas_from_center}), so they have zero brackets with all elements
of~$\Pi(\cA)$.  We therefore do not add these generators to the table of Poisson brackets and can provide a table
which is valid for both $M$ even and $M$ odd; in the table, we use again the abbreviations $I=(i_1,\dots,i_{N})$
and $J=(j_1,\dots,j_N)$, and $1\leqslant k,\ell\leqslant N$ while $1\leqslant i,j\leqslant M$.

{\small \begin{table}[h]
  \def\arraystretch{1.8}
  \setlength\tabcolsep{0.4cm}
  \centering
\begin{tabular}{c|ccc}
   $\PB$&$X_\ell$&$Y_j$&$W_J$\\
  \hline
  $X_k$&$4\(\delta_{k,\ell+1}-\delta_{k,\ell-1}\)X_k\cdot X_\ell$&
  $2\(\delta_{j,\left\lbrack\frac{k-1}2\right\rbrack}-\delta_{j,\left\lbrack\frac{k+1}2\right\rbrack}\)X_k\cdot Y_j$
    &$2\({j_{k-1}}-{j_{k+1}}\)X_k\cdot W_J$\\
  $Y_i$&$2\(\delta_{i,\left\lbrack\frac{\ell+1}2\right\rbrack}-\delta_{i,\left\lbrack\frac{\ell-1}2\right\rbrack}\)X_i\cdot Y_\ell$
      &$\(\delta_{i,j+1}-\delta_{i,j-1}\)Y_i\cdot Y_j$&$\epsilon(J,i)Y_i\cdot W_J$\\
  $W_I$&$2\(i_{\ell+1}-i_{\ell-1}\)X_\ell\cdot W_I$&$-\epsilon(I,j)Y_j W_I$& $\((-1)^{\eta(I,J)}-(-1)^{\eta(J,I)}\)W_{I+J}$
\end{tabular}
\bigskip
\caption{}\label{tab:vol_6}
\end{table}}

\noindent In this table, $\epsilon(I,\ell):={i_{2\ell-1}}+i_{2\ell}-{i_{2\ell+1}}-i_{2\ell+2}$, and the exponent
$\eta(I,J)$, which can be reduced modulo $2$, is defined by the equality
\begin{equation*}
  (x_1^{i_1}\dots x_{N}^{i_{N}})(x_1^{j_1}\dots x_{N}^{j_{N}})=(-1)^{\eta(I,J)}\;x_1^{i_1+j_1}\dots
  x_{N}^{i_{N}+j_{N}}\;,
\end{equation*}
where the product in the left-hand side is the product in $\cA$. An explicit formula is computed from it using the
commutation relations \eqref{eq:comm_A} and is given by
$\eta(I,J)=\sum_{s=1}^{M}(i_{2s}+i_{2s+1})\sum_{t=1}^{2s-1}j_t$, where we have set $i_{N+1}:=0$. Notice that again
$W_{i_1+j_1,\dots,i_{N}+j_{N}}$ may not belong to the chosen generators of $\Pi(\cA)$ (for example when one of the
indices becomes $2$), but is then easily written as a product of such generators.

For future use, we also give the $\Pi(\cA)$-Poisson module structure of $\cA$ in terms of the generators
$X_1,\dots,X_N,Y_1,\dots,Y_M$ of $\Pi(\cA)$ (to which the Casimirs $Y$ and $Z$, which belong to the annihilator of
$\cA$ (see Proposition~\ref{prp:Cas_from_center}), need to be added when $M$ is odd) and the generators
$x_1,\dots,x_N$ of $\cA$. Let $1\leqslant k,\ell\leqslant N$, $1\leqslant i\leqslant M$ and
$I=(i_1,\dots,i_N)$. Then $X_k\cdot x_\ell=x_k^2x_\ell$, $Y_i\cdot x_\ell=x_{2i}x_{2i+1}x_\ell$ and $W_I\cdot
x_\ell=0$, while
\begin{align}
  \pbm{X_k}{x_\ell}&=\pb{x_k^2,x_\ell}_1=2(\delta_{k,\ell+1}-\delta_{k,\ell-1})x_k^2x_\ell\;,\label{eq:mod_volt_3}\\
  \pbm{Y_i}{x_\ell}&=\pb{x_{2i}x_{2i+1},x_\ell}_1=\(\delta_{i,\left\lbrack\frac{\ell+1}2\right\rbrack}-
        \delta_{i,\left\lbrack\frac{\ell-1}2\right\rbrack}\)x_{2i}x_{2i+1}x_\ell\;,\label{eq:mod_volt_3b}
\end{align}
and
\begin{align}\label{eq:mod_volt_4}
  \pbm{W_I}{x_\ell}=\lb{x_1^{i_1}\dots x_N^{i_N},x_\ell}
  =\left\{
  \begin{array}{rl}
    \((-1)^{i_{\ell+2}+\cdots+i_N}-(-1)^{i_1+\cdots+i_{\ell-1}}\)x_1^{i_1}\dots x_\ell^{i_\ell+1}\dots x_N^{i_N}&\ell\hbox{ even},\\
    \((-1)^{i_{\ell+1}+\cdots+i_N}-(-1)^{i_1+\cdots+i_{\ell-2}}\)x_1^{i_1}\dots x_\ell^{i_\ell+1}\dots x_N^{i_N}&\ell\hbox{ odd}.
  \end{array}
  \right.
\end{align}
It is understood that when for example $\ell=N$ then the sums $i_{\ell+1}+\cdots+i_N$ and $i_{\ell+2}+\cdots+i_N$
in \eqref{eq:mod_volt_4} have no terms, hence are equal to 0.
\begin{remark}
  A similar remark as Remark \ref{rem:non_periodic} applies here: the formulas for the infinite case are the same,
  only the formula for $\eta(I,J)$ needs to be adapted, and the Casimirs from the periodic case become infinite
  products to be dealt with appropriately.
\end{remark}

\subsection{A quantisation of the Grassmann algebra}\label{par:Grassmann}
We now consider the deformation of the (complex two-dimensional) Grassmann algebra, given by the following
commutation relations:
\begin{equation*}
  \begin{array}{ll}
    \lb{x,p}=\h\;,&\psi^2=(\phi)^2=0\;,\\
    \psi\phi+\phi\psi=\h\;,\quad&\lb{p,\psi}=\lb{x,\psi}=\lb{p,\phi}=\lb{x,\phi}=0\;.
  \end{array}
\end{equation*}
We will only be interested in the evaluation for $\h=0$, so we can consider right away the algebra
\begin{equation*}
  \cA_\h=\frac{\bbC[[\h]]\langle p,x,\psi,\phi\rangle}{\cI_{\h}}\simeq\cA[[\h]]\;,\hbox{ where }
  \cA:=\frac{\bbC\langle p,x,\psi,\phi\rangle}{\cI_0}\;,
\end{equation*}
and where $\cI_\h$ stands for the closed ideal of $\bbC[[\h]]$, generated by $\lb{x,p}-\h,\ \psi^2,\ (\phi)^2,$
$\psi\phi+\phi\psi-\h,\ \lb{p,\psi},\ \lb{x,\psi}$, $\lb{p,\phi},\ \lb{x,\phi},$ and $\cI_0$ stands for its
evaluation at $\h=0$. As a monomial basis for $\cA_\h$ we take the monomials $p^ix^j\psi^k(\phi)^\ell$ with
$i,j\in\bbN$ and $k,\ell\in\set{0,1}$. In view of the latter restrictions on $k$ and $\ell$ it is not a PBW
basis. Since in $\cA$ the commutator of two elements is proportional to their product, its center is generated by
monomials, from which it is clear that $Z(\cA)$ is generated by $p,\; x$ and $\psi\phi$ and that $\cA/Z(\cA)$ is
generated as a $Z(\cA)$-module by $\cl \psi$ and $\cl{\phi}$. According to Proposition \ref{prop:generators}, it
follows that $\Pi(\cA)$ is generated by
\begin{equation*}
  P:=\(p,\cl0\)\;,\quad  X:=\(x,\cl0\)\;,\quad  W:=\(\psi\phi,\cl{0}\)\;,\quad \Psi:=\(0,\cl \psi\)\;,\quad
  \Phi:=\(0,\cl{\phi}\)\;.
\end{equation*}
To determine the algebra structure of $\Pi(\cA)$ we consider the surjective morphism
$\bbC[P,X,W,\Psi,\Phi]\to\Pi(\cA)$, defined by these generators. Let $\cJ:=\langle
\Psi^2,(\Phi)^2,W^2,\Psi\cdot\Phi,\Psi\cdot W,\Phi\cdot W\rangle$. It is clear that~$\cJ$ is contained in the
kernel of this morphism; to show that it coincides with the kernel, it suffices to notice that
\begin{gather}
  P^i\cdot X^j=(p^ix^j,\cl0)\;,\qquad  P^i\cdot X^j\cdot W=(p^ix^j\psi\phi,\cl0)\;, \label{eq:gras_1} \\
  P^i\cdot X^j\cdot \Psi=(0,\cl{p^ix^j\psi})\;,\qquad  P^i\cdot X^j\cdot \Phi=(0,\cl{p^ix^j\phi})\;,
\end{gather}
which shows that the family $\{P^i\cdot X^j,\, P^i\cdot X^j\cdot W,\, P^i\cdot X^j\cdot \Psi,$ $P^i\cdot
X^j\cdot \Phi\mid i,j\in\bbN\}$ is linearly independent. To compute the two products in \eqref{eq:gras_1}, we have
used that $(p^i,x^j)_1=0$ and $(p^ix^j,\psi\phi)_1=0$, for $i,j\in\bbN$. It follows that
$\(\Pi(\cA),\cdot\)\simeq\bbC[P,X,W,\Psi,\Phi]/\cJ$.
%
%
In order to compute the Poisson brackets between the generators, we use the following non-trivial brackets:
$\pb{\phi,\psi}_1=1$, $\pb{x,p}_1=1$, $\pb{\psi\phi,\psi}_1=\psi$ and $\pb{\psi\phi,\phi}_1=-\phi$. From these formulas and \eqref{eq:bra_1} --
\eqref{eq:bra_3} the Poisson brackets between the generators of~$\Pi(\cA)$ are easily computed. They are given in
Table \ref{tab:gras_1}.
\begin{table}[h]
  \def\arraystretch{1.8}
  \setlength\tabcolsep{0.4cm}
  \centering
\begin{tabular}{c|ccccc}
   $\PB$&$P$&$X$&$W$&$\Psi$&$\Phi$\\
  \hline
  $P$&$0$&${-1}$&$0$&$0$&$0$\\
  $X$&${1}$&$0$&$0$&$0$&$0$\\
  $W$&$0$&$0$&$0$&$\Psi$&$-\Phi$\\
  $\Psi$&$0$&$0$&$-\Psi$&$0$&$0$\\
  $\Phi$&$0$&$0$&$\Phi$&$0$&$0$\\
\end{tabular}
\smallskip
\caption{}\label{tab:gras_1}
\end{table}

\noindent
For example
\begin{equation*}
  \pb{P,X}=\pb{(p,\cl0),(x,\cl0)}=(\pb{p,x}_1,\cl0)=({-1},\cl0)={-1}\;,
\end{equation*}
and
\begin{equation*}
  \pb{W,\Psi}=\pb{({\psi\phi},0),(0,\cl \psi)}=\(0,\cl{\pb{\psi\phi,\psi}_1}\)=\(0,\cl \psi\)=\Psi\;.
\end{equation*}
It follows that the reduced Poisson algebra $Z(\cA)$ is isomorphic to the polynomial algebra $\bbC[P,X,W]$ with $W$
as Casimir and $\pb{X,P}=1$. The Poisson module structure on $\cA$ is computed similarly and is given in Table
\ref{tab:gras_2}.

\begin{table}[h]
  \def\arraystretch{1.8}
  \setlength\tabcolsep{0.4cm}
\begin{tabular}{c|ccccc}
   $\cdot$&$p$&$x$&$\psi$&$\phi$\\
  \hline
  $P$&$p^2$&$px$&$p\psi$&$p\phi$\\
  $X$&$px$&$x^2$&$x\psi$&$x\phi$\\
  $W$&$p\psi\phi$&$x\psi\phi$&$0$&$0$\\
  $\Psi$&$0$&$0$&$0$&$0$\\
  $\Phi$&$0$&$0$&$0$&$0$\\
\end{tabular}
\qquad
\begin{tabular}{c|ccccc}
   $\PBm$&$p$&$x$&$\psi$&$\phi$\\
  \hline
  $P$&$0$&${-1}$&$0$&$0$\\
  $X$&${1}$&$0$&$0$&$0$\\
  $W$&$0$&$0$&$\psi$&$-\phi$\\
  $\Psi$&$0$&$0$&$0$&$2\psi\phi$\\
  $\Phi$&$0$&$0$&$2\psi\phi$&$0$\\
\end{tabular}
\medskip
\caption{}\label{tab:gras_2}
\end{table}
%
%
\subsection{The algebra $M_q(2)$}\label{par:M_q2}
The entries of the $2\times2$-matrices $\begin{pmatrix}a&b\\c&d\end{pmatrix}$ which preserve the quantum plane
$\bbC_q[x,y]$ satisfy the following relations (see \cite[Ch.\ 4]{kassel}):
\begin{equation*}\setlength\arraycolsep{20pt}
   \begin{array}{ccc}
       ba=qab\;, &db=qbd\;,&bc=cb\;,\\
       dc=qcd\;, &ca=qac\;,&\ \ \ \ \ ad-da=(q^{-1}-q)bc\;.
   \end{array}
\end{equation*}
We are therefore led to consider the following quotient algebra:
\begin{equation*}
  \cA_q=\frac{\bbC(q)\langle a,b,c,d\rangle}{\cI_q}\;,
\end{equation*}
where $\cI_q$ is the ideal of $\bbC(q)$ generated by the following polynomials:
\begin{equation}\label{eq:ideal_M2}
  ba-qab\;,db-qbd\;, bc-cb\;,dc-qcd\;,ca-qac\;,ad-da-(q^{-1}-q)bc\;.
\end{equation}
It is well-known (and easy to check) that $ad-q^{-1}bc=da-qbc$ is a central element of $\cA_q$.  As a
monomial basis for $\cA_q$ we choose the monomials $a^ib^jc^kd^\ell$, with $i,j,k$ and $\ell$ in $\bbN$.  It is a
PBW basis and $q_0=0$ is the only non-regular value of $\cB$.

We will first consider the evaluation of $\cA_q$ to $q=1$, which is commutative. We do it to show how the classical
case of deformations of commutative algebras is treated as a special case of our methods. We set $q(\h)=1-\h$ and
consider
\begin{equation*}
  \frac{\bbC[[\h]]\langle a,b,c,d\rangle}{\cI_{q(\h)}}\simeq\cA[[\h]]\;,\hbox{ where }
  \cA=\frac{\bbC\langle a,b,c,d\rangle}{\cI_1}\simeq\bbC[a,b,c,d]\;.
\end{equation*}
In this case, $Z(\cA)=\cA$ and we may identify $(\Pi(\cA),\cdot)$ with the polynomial algebra $\cA$. Under this
identification the Poisson algebra $(\Pi(\cA),\LB_\h)$ and the reduced Poisson algebra $(\cA,\PB_1)$ coincide and
only the brackets $\PB_1$ between the elements $a,b,c,d$ need to be computed. These brackets are given in Table
\ref{tab:M2_1}.

\begin{table}[h]
  \def\arraystretch{1.8}
  \setlength\tabcolsep{0.5cm}
  \centering
\begin{tabular}{c|ccccc}
   $\PB=\PB_1=\PBm$&$a$&$b$&$c$&$d$\\
  \hline
  $a$&$0$&$-ab$&$-ac$&$-2bc$\\
  $b$&$ab$&$0$&$0$&$-bd$\\
  $c$&$ac$&$0$&0&$-cd$\\
  $d$&$2bc$&$bd$&$cd$&0\\
\end{tabular}
\bigskip
\caption{}\label{tab:M2_1}
\end{table}

For example, $\pb{a,d}=\pb{a,d}_1=-2bc$ since
\begin{equation*}
  \lb{a,d}=ad-da=\(q^{-1}-q\)bc=((1-\h)-(1+\h))bc+\cO\(\h^2\)=-2\h bc+\cO\(\h^2\)\;.
\end{equation*}
This Poisson structure has rank 2. Since $ad-q^{-1}bc$ is a central element of $M_q(2)$, $ad-bc$ is a
Casimir. Since $b/c$ as also a (rational) Casimir, the Poisson structure can be described as a Nambu-Poisson
structure (see \cite[Ch.\ 8.3]{PLV}).

We now consider the case $q=-1$. We set $q(\h)=\h-1$ and consider
\begin{equation*}
  \frac{\bbC[[\h]]\langle a,b,c,d\rangle}{\cI_{q(\h)}}\simeq\cA[[\h]]\;,\hbox{ where }
  \cA=\frac{\bbC\langle a,b,c,d\rangle}{\cI_{-1}}\;.
\end{equation*}
Notice that  $\cA$ is not commutative since the following relations hold in $\cA$:
\begin{equation}\setlength\arraycolsep{20pt}\label{eq:com_mq}
   \begin{array}{ccc}
       ba=-ab\;, &db=-bd\;,&bc=cb\;,\\
       dc=-cd\;, &ca=-ac\;,&ad=da\;.
   \end{array}
\end{equation}
It is clear from these relations that $Z(\cA)$ contains the following $6$ elements:
\begin{equation}\label{eq:gen_center}
  a^2,\;  b^2,\;  c^2,\;  d^2,\;  ad,\;  bc\;.
\end{equation}
We show that these elements generate $Z(\cA)$. To do this, first notice that $Z(\cA)$ is generated by monomials,
since the commutator of any two monomials is a multiple of their product, as follows from \eqref{eq:com_mq}. Given
a monomial in $Z(\cA)$ we may strip off any even power of $a$, $b$, $c$ and $d$, which leaves us with a monomial of
degree 1 at most in each of these variables. Among the 15 possible monomials of this type, it is easily checked
that only $ad$, $bc$ and their product $abcd$ belong to the center. This shows that the six elements
\eqref{eq:gen_center} generate~$Z(\cA)$. It is then clear that $\cA/Z(\cA)$ is generated as a $Z(\cA)$-module by
the following $8$ elements: $\cl a,\dots,\cl d,\cl{ab},\cl{ac},\cl{bd},\cl{cd}$.  According to Proposition
\ref{prop:generators}, it follows that the following elements generate $(\Pi(\cA),\cdot)$:
\begin{gather}
  U_1=\(a^2,\cl0\)\;,U_2=\(b^2,\cl0\)\;,U_3=\(c^2,\cl0\)\;,U_4=\(d^2,\cl0\)\;,U_5=\(ad,\cl0\)\;,U_6=\(bc,\cl0\)\;,
  \nonumber\\
  V_1=\(0,\cl a\)\;,\  V_2=\(0,\cl b\)\;,\  V_3=\(0,\cl c\)\;,\  V_4=\(0,\cl d\)\;,\label{eq:geners}\\
  W_1=\(0,\cl{ab}\)\;,\  W_2=\(0,\cl{ac}\)\;,\  W_3=\(0,\cl{bd}\)\;,\  W_4=\(0,\cl{cd}\)\;.\nonumber
\end{gather}

We now have all elements needed to determine the algebra structure of $(\Pi(\cA),\cdot)$, to compute the Poisson
brackets of the 14 generators $U_1,\dots,W_4$ of $\Pi(\cA)$ and to compute the Lie action of these generators on
the four generators $a,b,c,d$ of $\cA$. This can be done as in the previous examples by setting $q(\h)=\h-1$ and
computing the first terms of the commutators $\LB_\star$ and $\LB$ and using \eqref{eq:bra_1} -- \eqref{eq:bra_3}
and \eqref{eq:mod_simp}. We will only describe here the reduced Poisson algebra $Z(\cA)$. To do this, let us write
$U_i=(u_i,\cl0)$ for $i=1,\dots,6$, so that $Z(\cA)$ is generated by $u_1,\dots,u_6$: as an associative algebra, it
is clear that $Z(\cA)\simeq\bbC[u_1,\dots,u_6]/\langle u_5^2-u_1u_4,u_6^2-u_2u_3\rangle$. The reduced brackets,
which are the brackets $\PB_1$ of $(Z(\cA),\PB_1)$, are given by Table \ref{tab:M2}.
\begin{table}[h]
  \def\arraystretch{1.8}
  \setlength\tabcolsep{0.2cm}
  \centering
\begin{tabular}{c|cccccc}
   $\PB_1$&$u_1$&$u_2$&$u_3$&$u_4$&$u_5$&$u_6$\\
  \hline
  $u_1$&$0$&$4u_1u_2$&$4u_1u_3$&$-8u_5u_6$&$-4u_1u_6$&$4u_1u_6$\\
  $u_2$&$-4u_1u_2$&$0$&$0$&$4u_2u_4$&$0$&$0$\\
  $u_3$&$-4u_1u_3$&$0$&0&$4u_3u_4$&$0$&$0$\\
  $u_4$&$8u_5u_6$&$-4u_2u_4$&$-4u_3u_4$&0&$4u_4u_6$&$-4u_4u_6$\\
  $u_5$&$4u_1u_6$&$0$&$0$&$-4u_4u_6$&$0$&$0$\\
  $u_6$&$-4u_1u_6$&$0$&$0$&$4u_4u_6$&$0$&$0$\\
\end{tabular}
\bigskip
\caption{}\label{tab:M2}
\end{table}

For example, $\pb{u_1,u_5}_1=-4u_1u_6$ since $\pb{a^2,ad}_1=-4a^2bc$, as follows from
\begin{align*}
  \lb{a^2,ad}&=a^3d-ada^2=a^3d-a(ad+2\h bc)a+\cO\(\h^2\)\\
  &=a^3d-a^2(ad+2\h bc)-2\h a^2bc+\cO\(\h^2\)=-4\h a^2bc+\cO\(\h^2\)\;.
\end{align*}

Since $ad-q^{-1}bc$ is a central element of $\cA_q$, $(ad+bc,\cl{bc/2})=(ad+bc,\cl0)=U_5+U_6$ is a Casimir
of~$\Pi(\cA)$. This is reflected in Table \ref{tab:M2} and it reduces the number of Poisson brackets to be
computed, since $\Pb{U_5}=-\Pb{U_6}$, so that $\Pb{u_5}_1=-\Pb{u_6}_1$.

On $\bbC^6$, with coordinates $u_1,\dots,u_6$, Table \ref{tab:M2} gives the Poisson matrix of a Poisson structure
of rank $2$, with Casimirs $C_1=u_5^2-u_1u_4$, $C_2=u_6^2-u_2u_3$, $C_3=u_2/u_3$ and $C_4=u_5+u_6$. The Poisson
structure on $Z(\cA)$ can therefore also be described as the Nambu-Poisson structure on $\bbC^6$ with respect to
these Casimirs, restricted to the subvariety defined by $C_1=C_2=0$.

\subsection{Nonequivalent deformations}

As it turns out, we have encountered in the above examples three isomorphic algebras $\cA$ and three deformations
of it. Indeed, taking $N=4$ and $q_0=-1$ in \eqref{eq:quantum_volterra}, $\cI_{q_0}=\langle
x_{i+1}x_i+x_ix_{i+1},x_ix_j-x_jx_i\rangle_{d_N(i,j)\neq1},$ so that $\cA:=\bbC\langle
x_1,\dots,x_4\rangle/\cI_{q_0}$ is defined by
\begin{equation*}\setlength\arraycolsep{20pt}
   \begin{array}{lll}
       x_2x_1=-x_1x_2\;, &x_3x_2=-x_2x_3\;,&x_4x_2=x_2x_4\;,\\
       x_4x_3=-x_3x_4\;, &x_4x_1=-x_1x_4\;,&x_3x_1=x_1x_3\;.
   \end{array}
\end{equation*}
These relations are the same relations as \eqref{eq:com_mq} under the correspondence $a\leftrightarrow x_1$,
$b\leftrightarrow x_2$, $c\leftrightarrow x_4$ and $d\leftrightarrow x_3$. We will use the reduced Poisson algebra
to show that the two corresponding deformations are not equivalent, as an application of Corollary
\ref{cor:equivalent}. See Remark \ref{rem:third_deformation} below for a third deformation of $\cA$. In order to
compare the two Poisson algebras it will we usefull to use the same notation for the generators of $Z(\cA)$, so we
will use $U_1,\dots,U_6$, as in \eqref{eq:geners}, which means that
\begin{equation*}
\begin{array}{lll}
  U_1=X_1=(x_1^2,\cl0)\;,  &U_2=X_2=(x_2^2,\cl0)\;,  &U_3=X_4=(x_4^2,\cl0)\;,\\
  U_4=X_3=(x_3^2,\cl0)\;,  &U_5=Y=(x_1x_3,\cl0)\;,  &U_6=Z=(x_2x_4,\cl0)\;.
\end{array}
\end{equation*}
We will also again write $U_i=(u_i,\cl0)$ for $i=1,\dots,6$, so that $u_1,\dots,u_N$ generates $Z(\cA)$.  As we
have seen in Section \ref{par:volterra_1}, $Y$ and $Z$ are Casimirs of $\Pi(\cA)$ and the non-zero brackets
between the $X_i$ are given by $\pb{X_k,X_{k\pm1}}=\pm 2X_kX_{k\pm1}$. In terms of $u_1,\dots,u_4$ it leads to the
Poisson brackets, given in Table \ref{tab:noneq}; we did not add the brackets with $u_5$ and $u_6$ to the table
because they are Casimirs.
\begin{table}[h]
  \def\arraystretch{1.8}
  \setlength\tabcolsep{0.2cm}
  \centering
\begin{tabular}{c|cccccc}
   $\PB'=\PB_1'$&$u_1$&$u_2$&$u_3$&$u_4$\\
  \hline
  $u_1$&$0$&$4u_1u_2$&$-4u_1u_3$&$0$\\
  $u_2$&$-4u_1u_2$&$0$&$0$&$4u_2u_4$\\
  $u_3$&$4u_1u_3$&$0$&0&$-4u_3u_4$\\
  $u_4$&$0$&$-4u_2u_4$&$4u_3u_4$&$0$\\
\end{tabular}
\bigskip
\caption{}\label{tab:noneq}
\end{table}

\noindent
We denote this Poisson structure on $Z(\cA)$ by $\PB_1'$.  In order to compare the Poisson structures $\PB_1$ and
$\PB_1'$ on $Z(\cA)$ we look at their singular locus, which is by definition the locus where the rank of the
Poisson structure drops. In our examples, the rank is two so the singular locus consists of the points where the
rank is zero, which amounts to considering in both cases the ideal generated by the entries of the table. We are
therefore led to consider the following two Poisson ideals of $Z(\cA)$:
\begin{align*}
  \cJ&:=\langle u_1u_2,u_1u_3,u_2u_4,u_3u_4,u_1u_6,u_4u_6,u_5u_6\rangle\;,\\
  \cJ'&:=\langle u_1u_2,u_1u_3,u_2u_4,u_3u_4\rangle\;.
\end{align*}
It is clear that $\cJ'$ is strictly contained in $\cJ$. Moreover, both ideals have the same radical, since
\begin{align*}
  (u_1u_6)^2&=u_1^2u_2u_3=(u_1u_2)(u_1u_3)\in\cJ'\;,\\
  (u_4u_6)^2&=u_4^2u_2u_3=(u_2u_4)(u_3u_4)\in\cJ'\;,\\
  (u_5u_6)^2&=u_1u_2u_3u_4\in\cJ'\;,
\end{align*}
where we have used several times that in $\cA$ the following relations hold: $u_2u_3=u_6^2$ and $u_1u_4=u_5^2$. It
follows that the two reduced Poisson algebras are not isomorphic. In view of Corollary \ref{cor:equivalent}, the
two deformations $(\cA[[\h]],\star)$ and $(\cA[[\h]],\star')$ of $\cA$ are non-equivalent.

\begin{remark}\label{rem:third_deformation}
Taking $n=4$ in Section \ref{eq:volterra_second} we get again the same algebra $\cA$, since we find the same
relations as \eqref{eq:com_mq} under the correspondence $a\leftrightarrow x_1$, $d\leftrightarrow x_2$,
$c\leftrightarrow x_3$ and $b\leftrightarrow x_4$. It can be shown using the same methods that the deformation of
$\cA$ is in this case again non-equivalent to the two other deformations of $\cA$ of which we have shown in this
section that they are non-equivalent.
\end{remark}

\section{Hamiltonian derivations}\label{sec:ex2}
Most of the quantum algebras that we have considered in the previous section appear in the literature as quantum
algebras for some non-trivial derivation, defining a nonabelian system on a free algebra. These derivations can in
general be written as Heisenberg derivations having non-trivial limits for certain values of the deformation
parameter $q$. We will show in this section that the corresponding limiting derivations are Hamiltonian derivations
with respect to the commutative Poisson algebra $(\Pi(\cA),\PB)$ that we have introduced in Section
\ref{sec:poisson}, and that they can be easily computed using the formulas for the actions of $\Pi(\cA)$ on the
Poisson module $\cA$ that we have computed for these examples in Section \ref{sec:ex1}. As in the previous section
we take $R=\bbC$ as our base ring.

\subsection{The limiting procedure}\label{par:limit}

In the paragraphs which follow we will apply the above results to several nonabelian chains. We outline the
procedure in separate steps.

\underline{Step 1} \quad Start from a nonabelian system (derivation) $\p:\fA\to\fA$ on a free associative algebra
$\fA=\bbC\langle{x_1,\;x_2,\;\dots\rangle}$. Many interesting such systems are known \cite{CMW22}. Often, they are
\emph{evolutionary}, which means that the derivation $\p$ is invariant under the shift $x_\ell\mapsto x_{\ell+k}$,
for a fixed $k$. In our case there will be a finite number of variables $x_1,\dots,x_N$ since the index $\ell$ of
$x_\ell$ is taken modulo $N$, so it still makes sense for $\p$ to be evolutionary.

\underline{Step 2} \quad Choose a quantisation ideal $\cI_q$ of $(\fA(q),\p)$, depending on a single parameter $q$,
where we recall that $\fA(q)=\bbC(q)\langle{x_1,\;x_2,\;\dots\rangle}$. Many such ideals are known
\cite{CMW22}. Denote by $\cB$ a basis of normally ordered monomials of the quantum $\bbC(q)$-algebra
$\fA(q)/\cI_q$. On this algebra, $\p$ induces a derivation which may be evolutionary or not, depending on whether
or not $\cI_q$ is invariant under the shift $x_\ell\mapsto x_{\ell+k}$ for fixed $k$.

\underline{Step 3} \quad Write the equation for $\p$ on $\cA_q=\fA(q)/\cI_q$ in the Heisenberg form
\begin{equation}\label{eq:comm_form}
  \p a=\frac1{\lambda(q)}\lb{\fH(q),a}\;,
\end{equation}
where $a\in\cA_\q$.  This is a non-trivial task but again many examples have been written in this form. In this
formula, $\fH(q)\in\cA_q$; it may be assumed that $\fH(q)$ and $\lambda(q)$ are polynomials in $q$ and have no
common non-constant factor.

\underline{Step 4} \quad We can specialize $q$ to any regular value $q_0$, but as we will see the most interesting
choice for $q_0$ is to choose a simple root of $\lambda(q)$. As before, we write $q(\h)=q_0+\h$. We recall that
$\cI_\h$ stands for the closed ideal of $\fA[[\h]]$ that corresponds with $\cI_{q(\h)}$. As we have shown in
Proposition \ref{prp:iso_to_formal}, $\cA_\h:=\fA[[\h]]/\cI_\h\simeq\cA[[\h]]$, where $\cA$ is the evaluation of
$\cA_q$ at $q_0$. Since $q_0$ is a simple root of $\lambda(q)$,
$\frac{\fH(q_0+\h)}{\lambda(q_0+\h)}\in\frac1\h\cA[[\h]]$ and \eqref{eq:comm_form} takes in terms of $\h$ the
Heisenberg form
\begin{equation}\label{eq:heisenberg}
  \delta_{H}a=\frac1\h\lb{H,a}_\star=\frac1\h\lb{H_0+\h H_1+\h^2H_2+\cdots,a}_\star\;,
\end{equation}
where $a,H\in\cA[[\h]]$, with ${H}=H_0+\h H_1+\h^2H_2+\cdots$, i.e., all $H_i$ are elements of $\cA$ which belong
to the $\bbC$-span of $\cB$. If $\p$ is evolutionary (on $\cA_\q$), then so is $\delta_H$ (on $\cA[[\h]]$).
Since the left-hand side of \eqref{eq:heisenberg} is a formal power series in $\h$, it follows that~$H_0$ commutes
with any element $a$ of $\cA\subset\cA[[\h]]$, hence $H_0\in Z(\cA)$ and $\bH:=\(H_0,\cl{H_1}\)\in\Pi(\cA)$.
Notice that, in order to determine $\bH$ it suffices to compute in \eqref{eq:heisenberg} the leading terms $H_0$
and $H_1$, the latter up to the center $Z(\cA)$ of $\cA$. We will therefore compute and write
\begin{equation}\label{eq:heisenberg_pre}
  \frac{\fH(q(\h))}{\lambda(q(\h))}=\frac1\h (H_0+\h H_1)\pmod{\cH_\h}\;,
\end{equation}
which suffices to determine $\bH$.

\underline{Step 5} \quad As we have shown in Section \ref{par:heisenberg}, the limit $\h\to0$ of the Heisenberg
derivation \eqref{eq:heisenberg} is the Hamiltonian derivation $\p_\bH$ on $\cA$, which can be computed directly
from $\p_\bH a=\pbm{\bH}a$, where we recall that $\PBm$ denotes the Lie action of $\Pi(\cA)$ on $\cA$. Let
$U_1,\dots,U_M$ denote a system of algebra generators of $\Pi(\cA)$ (recall that, as an algebra, $\Pi(\cA)$ is
commutative). We write $\bH$ in terms of the generators of $\Pi(\cA)$, $\bH=\bH(U_1,\dots,U_M)$. Since $\PBm$ is a
derivation in its first argument and since the left and right actions of $\Pi(\cA)$ on $\cA$ coincide (see
\eqref{eq:final_module}), $\p_\bH x_\ell$ can be computed from
\begin{equation}\label{eq:leibniz_for}
  \p_{\bH}x_\ell=\pbm\bH {x_\ell}=\sum_{i=1}^M\frac{\p\bH}{\p U_i}\cdot\pbm{U_i}{x_\ell}\;,
\end{equation}
where we recall that $\cdot$ denotes the left (= right) action of $\Pi(\cA)$ on $\cA$; notice that
\eqref{eq:final_module}) says in particular that the action $\cdot$ of $\{0\}\times\cA/Z(\cA)\subset\Pi(\cA)$ on
$\cA$ is trivial which permits to largely simplify the use of \eqref{eq:leibniz_for} in explicit computations. The
brackets $\pbm{U_i}{x_\ell}$ between the generators of $\Pi(\cA)$ and of $\cA$ have been computed for several
examples in Section \ref{sec:ex1}. Notice that the computations are done for the variables $x_\ell$ (and their
projections on the different quotient algebras), rather than for arbitrary elements $a$ of
$\fA=\bbC\langle{x_1,\;x_2,\;\dots\rangle}$ or of $\fA(q)=\bbC(q)\langle{x_1,\;x_2,\;\dots\rangle}$.  On the one
hand, it leads to simpler explicit formulas that are easier to compute and to present, while on the other hand
these formulas completely determine the derivation $\p_{\bH}$ of all of~$\cA$ (hence on $\cA[[\h]]$), because
$\p_{\bH}$ is a derivation of $\cA$ (Proposition \ref{prp:derivation}). Moreover, when $\delta_\bH$ is
evolutionary, say invariant for the shift $x_\ell\mapsto x_{\ell+k}$ then so is $\p_\bH$ and it suffices to compute
$\delta_\bH x_\ell$ for $k-1$ successive values of $\ell$ to know it for all $\ell$; as we will see this also
simplifies some of the computations.

\begin{remark}\label{hams}
  Suppose there exists an element $\hat{\fH}(q)\in\cA_q$ that commutes with $\fH(q)$ and is not in $Z(\cA(q))$. The
  corresponding derivation $\hat{\p} a = \frac1{\lambda(q)}\lb{\hat{\fH}(q),a}$ commutes with $\p$ and represents a
  symmetry for the quantum system (\ref{eq:comm_form}). Furthermore, the leading term in the expansion $
  \frac{\hat{\fH}(q(\h))}{\lambda(q(\h))}=\frac1\h (\hat{H}_0+\h \hat{H}_1)$ becomes a first integral of the
  Hamiltonian equation (\ref{eq:leibniz_for}). If $\hat{H}_0\ne 0$, then $\p_{\bH}(\hat{H}_0)=0$ (see Proposition
  \ref{prp:invol}). In the case where $\hat{H}_0= 0$, it follows that $ \p_{\bH}(\hat{H}_1)=0$. If $\hat{\fH}(q)\in
  Z(\cA_q)$, then $\hat{\bH}:=\(\hat{H}_0,\cl{\hat{H}_1}\)$ is a Casimir of the Poisson structure (see Proposition
  \ref{prp:Cas_from_center}), and therefore $\hat{H}_0$ is a first integral of (\ref{eq:leibniz_for}).
\end{remark}

\subsection{$N$-periodic nonabelian Volterra hierarchy}\label{par:volterra}

The \emph{nonabelian Volterra chain} \cite{CMW22} is the derivation of $\fA=\langle{x_\ell\;;\ell\in\bbZ\rangle}$,
given by
\begin{equation}\label{eq:volterra_chain}
  \p_1 x_\ell=x_{\ell}x_{\ell+1}-x_{\ell-1}x_{\ell}\;.
\end{equation}
It has an infinite family of commuting derivations $\p_2,\p_3,\dots$, forming the so-called \emph{nonabelian
Volterra hierarchy}. We consider here the $N$-periodic case, that is $N\geqslant3$ and $x_{N+\ell}=x_\ell$ for
all~$\ell$. It was shown in \cite{CMW22} that all members of this hierarchy admit a common quantisation ideal,
namely the ideal $\cI_q$ of $\fA(q)$ generated by all
\begin{equation*}
  x_{i+1}x_i-qx_ix_{i+1}\;,\ x_i x_j-x_j x_i\;, \quad(d_N(i,j)\neq1)\;.
\end{equation*}
It is the quantisation ideal which we studied in Section \ref{par:volterra_1}, see in particular
\eqref{eq:quantum_volterra}.  Notice that $\cI_q$ is invariant under the shift $x_\ell\mapsto x_{\ell+1}$. We
determine some non-trivial limits of the (evolutionary) derivations $\p_n$ of the quantum algebra
$\cA_q=\fA(q)/\cI_q$. Recall that we use the monomials $x_1^{i_1}\dots x_N^{i_N}$ with $i_1,\dots,i_N\in\bbN$ as a
monomial basis for $\cA_q$.

To do this we use the results from \cite{CMW24}, where the full hierarchy on $\cA_q$ is written in the Heisenberg
form
\begin{equation}\label{eq:volterra_comm_form}
  \p_n x_\ell=\frac1{q^n-1}\lb{\fH^{(n)},x_\ell}\;,\qquad \ell=1,\dots,N\;,
\end{equation}
and where the Hamiltonians $\fH^{(n)}$ are given explicitly for any $n$. Here we will use the first three
Hamiltonians
\begin{align*}
  \fH^{(1)}&=\sum_{k=1}^N x_k\;,\\
  \fH^{(2)}&=\sum_{k=1}^N x_k^2+(1+q)\sum_{k=1}^N x_kx_{k+1}\;,\\
  \fH^{(3)}&=\sum_{k=1}^N x_k^3+(1+q+q^2)\sum_{k=1}^N\(x_kx_{k+1}x_{k+2}+x_kx_{k+1}^2+x_k^2x_{k+1}\)\;.
\end{align*}
As we already pointed out in Remark \ref{rem:n=1}, when $q\to1$, the algebra $\cA$ is commutative and the Poisson
brackets are given by $ \pb{x_i,x_j}=(\delta_{i,j+1}-\delta_{i,j-1})x_ix_j,$ for $1\leqslant i,j\leqslant N$. We
are then in the classical case, see Remark \ref{exa:class_poisson}.

We therefore start with the case of $n=2$ and $q=-1$, so we put $q(\h)=\xi+\h=-1+\h$. Then
%
\begin{equation*}
  \frac{\H2{}}{q(\h)^2-1}=-\frac1{2\h}\(\sum_{k=1}^N x_k^2+{\h}\sum_{k=1}^{N} {x_kx_{k+1}}\)\pmod{\cH_h}\;.
\end{equation*}
Setting $\bH^{(2)}=\(H^{(2)}_0,\cl{H^{(2)}_1}\)$, where $H^{(2)}_0=-\frac12\sum_{k=1}^N x_k^2$ and
$H^{(2)}_1=-\frac12\sum_{k=1}^{N}{x_kx_{k+1}}$ , we need to compute $\pbm{\bH^{(2)}}{x_\ell}$ for $\ell=1,\dots,N$,
which we will do for $N>4$. Since $\p_2$ is evolutionary (under $x_\ell\mapsto x_{\ell+1}$), it suffices to do the
computation for a particular $\ell$, so we may assume that $2<\ell<N-1$.  To do this we first write~$\bH^{(2)}$ in
terms of the generators $X_k$ and $W_I$ for $\Pi(\cA)$ that we have constructed in Section~\ref{par:volterra_1}. It
will be convenient to write $W_k$ as a shorthand for $W_{0,\dots,0,1,1,0,\dots,0}$ where the two $1$'s are at
positions~$k$ and $k+1$ (with the understanding that when $k=N$ then they are at positions $N$ and $1$). Then
$\bH^{(2)}=-\frac12\sum_{k=1}^N (X_k+W_k)+W_N$. Since $n=2$ and $\xi=-1$, the Lie action of $\Pi(\cA)$ on $\cA$,
given in \eqref{eq:mod_volt} and \eqref{eq:mod_volt_2}, specialises for $1<\ell<N$ to
\begin{align*}
  \pbm{X_k}{x_\ell}&=2(\delta_{k,\ell-1}-\delta_{k,\ell+1})x_k^2x_\ell\;, \\
  \pbm{W_k}{x_\ell}&=((-1)^{i_{\ell+1}}-(-1)^{i_{\ell-1}})x_1^{i_1}\dots x_\ell^{i_\ell+1}\dots x_N^{i_N}\;,
\end{align*}
where $i_k=i_{k+1}=1$ and all other $i_s$ are zero. Notice that these brackets are zero when $k$ and $\ell$ are far
enough apart, and also that $\pbm{X_{\ell}}{x_\ell}=0$. It follows that, if $2<\ell<N-1$ then
\begin{align}\nonumber
  \p_{\bH^{(2)}}x_\ell
    &=-(\pbm{X_{\ell-1}+X_{\ell+1}}{x_\ell}-\pbm{W_{\ell-2}+W_{\ell-1}+W_\ell+W_{\ell+1}}{x_\ell})/2\\ \nonumber
    &=(x_{\ell+1}^2-x_{\ell-1}^2)x_\ell-(x_{\ell-2}x_{\ell-1}x_\ell+x_{\ell-1}x_\ell^2-x_\ell^2x_{\ell+1}-x_\ell
    x_{\ell+1}x_{\ell+2})\\ \label{v2m1}
    &= x_\ell x_{\ell+1}(x_{\ell+2}+x_{\ell+1}-x_{\ell})-(x_{\ell-2}+x_{\ell-1}-x_{\ell})x_{\ell-1}x_\ell\;,
\end{align}
where we have used in the last step the commutation relations $x_{\ell+1}x_\ell=-x_{\ell}x_{\ell+1}$ to write the
result in a symmetric form. Since $\p_2$ is evolutionary, this formula is valid for all $\ell$.

When $n=1$, the denominator of \eqref{eq:volterra_comm_form} does not vanish at $q=-1$, and setting $q=\h-1$ leads
to the limiting derivation $\p_{\bH^{(1)}}$ on $\cA$, where $\bH^{(1)}=\(0,\cl{H_1^{(1)}}\)$, with
$H_1^{(1)}=\sum_{k=1}^Nx_k$. It is given by $\p_{\bH^{(1)}}x_\ell=x_\ell x_{\ell+1}-x_{\ell-1}x_\ell$ for
$\ell=1,\dots,N$. Since $\p_1$ and $\p_2$ commute, so do $\p_{\bH^{(1)}}$ and $\p_{\bH^{(2)}}$; this follows also
from the fact that $\pb{\bH^{(1)},\bH^{(2)}}=0$, see Proposition \ref{prp:invol}. The same remark applies to all odd derivations $\p_{2m+1}$ and their limiting derivations $\p_{\bH^{(2m+1)}}$ on $\cA$, where $m\in\bbN$. The Hamiltonian system (\ref{v2m1}) has first integrals $H_1^{(2k-1)},\ H_0^{(2k)},\ k=1,2,\ldots$ (see Remark \ref{hams}).

We now consider $n=3$, with $\xi$ a primitive cubic root of unity. As above we will only do this for $N>6$. We
still use the results of Section \ref{par:volterra_1} and put $q(\h)=\xi+\h$. Then
\begin{equation*}
\frac1{q(\h)^3-1}=\frac\xi{3\h}+\cO(1)\;,\quad\hbox{and}\quad\frac1{q(\h)-1}=\frac1{\xi-1}+\cO(\h)\;,
\end{equation*}
and
\begin{align*}
  \frac{\H3{}}{q(\h)^3-1}&=\frac\xi{3\h}\sum_{k=1}^Nx_k^3
  +\frac1{\xi-1}\sum_{k=1}^N(x_kx_{k+1}x_{k+2}+x_kx_{k+1}^2+x_k^2x_{k+1}) \pmod{\cH_\h}\;,\\
\end{align*}
so that $\bH^{(3)}=\(H_0^{(3)},H_1^{(3)}\)$, with
\begin{equation}\label{eq:V_n=3}
  H^{(3)}_0=\frac\xi3\sum_{k=1}^N x_k^3\;,\quad\hbox{and}\quad
  H^{(3)}_1=\frac1{\xi-1}\sum_{k=1}^N(x_kx_{k+1}x_{k+2}+x_kx_{k+1}^2+x_k^2x_{k+1})\;.
\end{equation}
It suffices again to compute $\pbm{\bH^{(3)}}{x_\ell}$ for a particular value of $\ell$ since $\p_3$ is
evolutionary (under $x_\ell\mapsto x_{\ell+1}$). In order to take care of the terms in $H_1^{(3)}$, let us write
$W_k$ as a shorthand for $W_{0,\dots,0,1,1,1,0,\dots,0}$, where the three 1's are at positions $k-1,k,k+1$, $W_k'$
as a shorthand for $W_{0,\dots,0,1,2,0,\dots,0}$ where the~1 is at position $k$ and $W_k''$ for
$W_{0,\dots,0,2,1,0,\dots,0}$ where the $2$ is at position $k$. In terms of this notation,~$\bH^{(3)}$ can be
written as
\begin{equation*}
  \bH^{(3)}=\frac\xi3\sum_{k=1}^N X_k+\frac1{\xi-1}\(\sum_{k=1}^{N-2}(W_{k+1}+W_k'+W_k'')+\xi^{-1}(W_1+W_N)+W'_{N-1}+W''_{N-1}+
  \xi(W_N'+W''_N)\)\;.
\end{equation*}
By our choice of $\ell$ we will manage that the 6 boundary terms which appear above play no role in the
computation.  In order to compute $\p_{\bH^{(3)}}x_\ell=\pbm{\bH^{(3)}}{x_\ell}$ we need the following brackets,
which are a specialization of \eqref{eq:mod_volt} and \eqref{eq:mod_volt_2} for $1<\ell<N$,
\begin{align}
  \pbm{X_k}{x_\ell}&=3(\delta_{k,\ell+1}-\delta_{k,\ell-1})\xi^{-1}x_k^3x_\ell\;, \label{eq:X_k}\\
  \pbm{W_k}{x_\ell}&=(\xi^{i_{\ell+1}}-\xi^{i_{\ell-1}})x_1^{i_1}\dots x_\ell^{i_\ell+1}\dots x_N^{i_N}\;,\label{eq:W_k}
\end{align}
where the latter formula is also valid for $W_k'$ and $W_k''$, each time upon using the proper values for the
indices $i_\ell$; for example, in the case of $W_k$, all indices are zero except $i_{k-1}=i_k=i_{k+1}=1$. In view
of \eqref{eq:X_k} and \eqref{eq:W_k}, $\pbm{X_\ell}{x_\ell}=\pbm{W_\ell}{x_\ell}=0.$ Also, the only non-zero
brackets $\pbm{W'_k}{x_\ell}$ and $\pbm{W''_k}{x_\ell}$ are for $k=\ell-2,\dots,\ell+1$, while the only non-zero
brackets $\pbm{W_k}{x_\ell}$ are for $k=\ell\pm1$ and $k=\ell\pm2$.  Let $3<\ell<N-2$, where we recall that
$N>6$. Then
\begin{align*}
  \p_{\bH^{(3)}}x_\ell&=\frac{\xi}3\pbm{X_{\ell-1}+X_{\ell+1}}{x_\ell}+
  \frac1{\xi-1}\pbm{W_{\ell-2}+W_{\ell-1}+W_{\ell+1}+W_{\ell+2}}{x_\ell}\\
  &\quad+\frac1{\xi-1}\pbm{W'_{\ell-2}+W'_{\ell-1}+W'_{\ell}+W'_{\ell+1}+W''_{\ell-2}+W''_{\ell-1}+W''_{\ell}+W''_{\ell+1}}
        {x_\ell}\\
  &=\(x_{\ell+1}^3-x_{\ell-1}^3\)x_\ell +(1+\xi)(x_\ell
    x_{\ell+1}^2x_{\ell+2}-x_{\ell-2}x_{\ell-1}^2x_\ell+x_\ell^2x_{\ell+1}^2-x_{\ell-1}^2x_\ell^2)\\
  &\quad +x_\ell x_{\ell+1}x_{\ell+2}x_{\ell+3}-x_{\ell-3}x_{\ell-2}x_{\ell-1}x_\ell+
    x_\ell^2x_{\ell+1}x_{\ell+2}-x_{\ell-2}x_{\ell-1}x_\ell^2\\
  &\quad+x_\ell x_{\ell+1}x_{\ell+2}^2-x_{\ell-2}^2x_{\ell-1}x_\ell+x_\ell^3x_{\ell+1}-x_{\ell-1}x_\ell^3\;.
\end{align*}
The formulas that we have computed for the limiting derivations are also valid for the infinite (non-periodic) case
are the same, with the same proof (see Remark \ref{rem:non_periodic}).

\subsection{$2M$-periodic even nonabelian Volterra hierarchy: another quantisation}\label{par:another}

We now consider another quantisation ideal of the even elements $\delta_{n}=\delta_{2m}$ of the $N$-periodic
nonabelian Volterra hierarchy, in case $N>2$ is even, $N=2M$ (see \cite{CMW22}). The ideal $\cI_q$ of $\fA(q)$ is
now generated by all
\begin{equation}\label{eq:sec_comm}
  x_{i+1}x_i-(-1)^iqx_ix_{i+1}\;, \ x_i x_j+x_j x_i\;, \quad(d_N(i,j)\neq1)\;.
\end{equation}
We have already considered this quantisation ideal, the corresponding quantum algebras $\cA_q=\fA(q)/\cI_q$ and its
evaluation at $q=1$ in Section \ref{eq:volterra_second}; we will use here the Poisson brackets from that section to
obtain the limiting derivation of $\p_2$ (see \eqref{eq:volterra_comm_form}) when $q\to1$. The derivation $\p_2$ is
given on $\cA_q$ in Heisenberg form by
\begin{equation*}
  \p_2 x_\ell=\frac1{q^2-1}\lb{\fH,x_\ell}=\frac1{q^2-1}\lb{\sum_{k=1}^N\(x_k^2+(1+(-1)^kq)x_{k} x_{k+1}\),x_\ell}\;.
\end{equation*}
We put, as in Section \ref{eq:volterra_second}, $q(\h)=1+\h$. Then
\begin{equation*}
  \frac1{q(\h)^2-1}=\frac1{2\h}+\cO(\h^0)\;,\hbox{ and }\  1+(-1)^kq(\h)=
    \begin{cases}
      2+\h&k\hbox{ even,}\\
      -\h&k\hbox{ odd.}
    \end{cases}
\end{equation*}
and
\begin{equation*}
  \frac{\fH}{q(\h)^2-1}=\frac1{2\h}\(\sum_{k=1}^Nx_k^2+\sum_{j=1}^Mx_{2j}x_{2j+1}-
  \frac\h2\sum_{j=1}^Mx_{2j-1}x_{2j}\)\pmod{\cH_\h}\;,
\end{equation*}
so that
\begin{equation*}
  H_0=\frac12\sum_{k=1}^Nx_k^2+\sum_{j=1}^Mx_{2j}x_{2j+1}\;\hbox{ and }
  H_1=-\frac12\sum_{j=1}^Mx_{2j-1}x_{2j}\;.
\end{equation*}
Let us write $W_k$ as a shorthand for $W_{0,\dots,0,1,1,0,\dots,0}$ where the two $1$'s are at positions~$2k-1$ and
$2k$. Then
$$
  \bH=\frac12\sum_{k=1}^NX_k+\sum_{j=1}^{M-1}Y_j-Y_M-\frac12\sum_{j=1}^MW_j\;,
$$
where $X_1,\dots,X_{N},Y_1,\dots,Y_M$ are the generators of $\Pi(\cA)$, constructed in Section
\ref{eq:volterra_second}.  We need to compute $\p_\bH x_\ell$, which we do for even $\ell$ only, the computation
for odd $\ell$ being very similar. Then the only non-zero brackets $\pbm{X_k}{x_\ell}$, $\pbm{Y_j}{x_\ell}$ and
$\pbm{W_k}{x_\ell}$ are according to \eqref{eq:mod_volt_3} -- \eqref{eq:mod_volt_4} given by
\begin{equation*}\renewcommand*{\arraystretch}{1.5}
\begin{array}{rclrcl}
  \pbm{X_{\ell-1}}{x_\ell}&=&-2x_{\ell-1}^2x_\ell\;,&  \pbm{X_{\ell+1}}{x_\ell}&=&2x_{\ell}x_{\ell+1}^2\;,\\
  \pbm{Y_{\ell/2-1}}{x_\ell}&=&-x_{\ell-2}x_{\ell-1}x_\ell\;,&  \pbm{Y_{\ell/2}}{x_\ell}&=&x_{\ell}^2x_{\ell+1}\;,\\
  \pbm{W_{\ell/2-1}}{x_\ell}&=&2x_{\ell-1}x_\ell^2\;,&  \pbm{W_{\ell/2+1}}{x_\ell}&=&-2x_{\ell}x_{\ell+1}x_{\ell+2}\;.
\end{array}
\end{equation*}
Suppose that $1<\ell<N$ (recall that $\ell$ and $N$ are even). Then
\begin{align*}
  \p_{\bH} x_\ell&=\frac12\pbm{X_{\ell-1}+X_{\ell+1}}{x_\ell}+\pbm{Y_{\ell/2-1}+Y_{\ell/2}}{x_\ell}-
  \frac12\pbm{W_{\ell/2-1}+W_{\ell/2}}{x_\ell}\;\\
  &=x_\ell x_{\ell+1}^2-x_{\ell-1}^2x_\ell+x_\ell^2x_{\ell+1}-x_{\ell-2}x_{\ell-1}x_\ell+x_{\ell}x_{\ell+1}x_{\ell+2}-
  x_{\ell-1}x_\ell^2\;.
\end{align*}
Since the Volterra hierarchy and the ideal are invariant under the shift $x_i\mapsto x_{i+2}$, the above formula is
valid also for $\ell=N$. It is in fact valid for all $\ell$, even in the infinite ($N=\infty$) case.


\subsection{A system on the Grassmann algebra}\label{par:Grassmann_2}
We now consider some simple dynamics on the Grassmann algebra, which we already considered, together with its
deformation in Section \ref{par:Grassmann}. On $\cA_\h\simeq\cA[[\h]]$, consider the derivation defined for
$a\in\cA$ by
\begin{equation*}
  \p a=\frac1\h\lb{\fH,a}\;,\quad\hbox{where}\quad \fH=\frac12(p^2+x^2)+x\psi\phi\;.
\end{equation*}
It is already written in the Heisenberg form and the corresponding Hamiltonian $\bH\in\Pi(\cA)$ is given by
$\bH=\frac12(P^2+X^2)+XW$. The limiting derivation, for $\h\to0$, is given by
\begin{equation*}
  \p_{\bH}\psi=x\psi\;,\quad  \p_{\bH}\phi=-x\phi\;,\quad  \p_{\bH}p=x+\psi\phi\;, \quad \p_{\bH}x=-p\;.
\end{equation*}
This is an easy consequence of \eqref{eq:leibniz_for}, upon using Table \ref{tab:gras_2}. For example,
\begin{align*}
  \p_{\bH}\psi&=\pbm{\frac12(P^2+X^2)+XW}\psi=P\cdot\pbm P\psi+(X+W)\cdot\pbm X\psi+X\cdot\pbm W\psi=X\cdot
  \psi=x\psi\;,\\
  \p_{\bH}p&=\pbm{\frac12(P^2+X^2)+XW}p=P\cdot\pbm Pp+(X+W)\cdot\pbm Xp+X\cdot\pbm
  Wp=(X+W)\cdot 1=x+\psi\phi\;.
\end{align*}

\subsection{A hierarchy on the quantum plane}\label{simp}

We now consider an example related to the quantum plane which we considered in Section \ref{par:quantum_plane}. On
the free algebra $\bbC\langle x,y\rangle$, there is a hierarchy of commuting derivations $\p_n$, $n>0$, defined
by
\begin{equation}\label{uv}
  \p_nx = xy(y-x)^{n-1},\quad \p_ny = yx(y-x)^{n-1}\;;
\end{equation}
see \cite{MS2000CMP}.  The ideal $\cI_q:=\langle yx-qxy\rangle$ of $\bbC\langle x,y\rangle$ is a quantisation ideal
for each one of these derivations, since $\p_n(yx-qxy)\in\langle yx-qxy\rangle$. The derivations $\p_n$ therefore
descend to commuting derivations of the quantum plane $\frac{\bbC(q)\langle x,y\rangle}{\langle yx-qxy\rangle}$,
and they can be written in the Heisenberg form
\begin{equation}\label{eq:comm_form_1}
 \p_nx = \frac{1}{q^n-1}[\fH_n,x]\;,\quad \p_ny = \frac{1}{q^n-1}[\fH_n,y]\;,
\end{equation}
where $\fH^{(n)}=(y-qx)^n$. Notice that, using the $q$-binomial formula \cite[Prop.\ IV.2.2]{kassel},
\begin{equation*}
  \fH^{(n)}=(y-qx)^n=\sum_{k=0}^n\binom n k_q(-qx)^ky^{n-k}
    =y^n+(-q)^nx^n+\sum_{k=1}^{n-1}\binom n k_q(-qx)^ky^{n-k}\;,
\end{equation*}
where the $q$-binomial coefficients are given by
\begin{equation*}
  \binom n k_q=\frac{(q^n-1)(q^{n-1}-1)\dots(q^{n-k+1}-1)}{(q-1)(q^2-1)\dots(q^k-1)}\;.
\end{equation*}
It is easily checked that they satisfy the following well-known recursion relation
\begin{equation}\label{eq:binom_rec}
  \binom n k_q=\frac{q^n-1}{q^{n-k}-1}\binom {n-1}k_q\;,
\end{equation}
which we will use.  We consider the limiting derivation of \eqref{eq:comm_form_1} for $q\to\xi$, where $\xi$ is a
primitive $n$-th root of unity. We therefore set, as in Section \ref{par:quantum_plane}, $q(\h)=\xi+\h$. Since
$q(\h)^n-1=\xi^{-1}n\h+\cO(\h^2)$, while $\binom {n-1}k_q=\binom {n-1}k_\xi+\cO(\h)$, where the constant term is
non-zero, we get, using \eqref{eq:binom_rec}, for $0<k<n$,
\begin{equation*}
  \binom n k_{q(\nu)}=\frac{\xi^{-1}n\h}{\xi^{n-k}-1}\binom {n-1}k_\xi+\cO(\h^2)\;
\end{equation*}
and
\begin{equation}\label{eq:afs}
  \frac{\fH^{(n)}}{q(\h)^n-1}=\frac\xi{n\h}(y^n+(-1)^{n}x^n)+\sum_{k=1}^{n-1}\frac{(-\xi)^k\binom{n-1}k_\xi}
       {\xi^{n-k}-1}x^ky^{n-k}\pmod{\cH_\h}\;,
\end{equation}
so that
\begin{equation*}
  H_0^{(n)}=\frac\xi n(y^n+(-1)^nx^n)\;,\hbox{ and }
  H_1^{(n)}=\sum_{k=1}^{n-1}\frac{(-\xi)^k}{\xi^{n-k}-1}\binom{n-1}k_\xi x^ky^{n-k}\;.
\end{equation*}
Recall that the center of $\cA$ is generated by $X=(x^n,\cl0)$, $Y=(y^n,\cl0)$ and $W_{i,j}=\(0,\cl{x^iy^j}\)$,
with $0\leqslant i,j\leqslant n,\ i+j\neq0$ (see Section \ref{par:quantum_plane}). Setting
$\bH^{(n)}=(H^{(n)}_0,H^{(n)}_1)$ we can write $\bH^{(n)}$ in terms of these generators as
\begin{align*}
  \bH^{(n)}=\frac\xi n(Y+(-1)^nX)+\sum_{k=1}^{n-1}\frac{(-\xi)^k}{\xi^{n-k}-1}\binom{n-1}k_\xi W_{k,n-k}\;.
\end{align*}
Using Table \ref{tab:intro_2} the limiting derivation $\p_{\bH^{(n)}}=\pbm{\bH^{(n)}}{\cdot}$ is given by
\begin{align*}
  \p_{\bH^{(n)}}x&=\pbm{\bH^{(n)}}x=xy^n+\sum_{k=1}^{n-1}(-1)^k\xi^k \binom{n-1}k_\xi x^{k+1}y^{n-k}\;,\\
  \p_{\bH^{(n)}}y&=\pbm{\bH^{(n)}}y=(-1)^{n-1}x^ny+\sum_{k=1}^{n-1}(-1)^k\xi^{2k} \binom{n-1}k_\xi x^{k}y^{n-k+1}\;.
\end{align*}


\subsection{A hierarchy on the quantum torus}\label{par:quantum_torus_2}

We now consider an example related to the quantum torus (see Section~\ref{par:quantum_plane}, especially Remark
\ref{rem:quantum_torus}). According to \cite{Wolf2012}, Kontsevich considered on the algebra $\fA:=\bbC\langle
x,y,x^{-1},y^{-1}\rangle$ the derivation, defined by
\begin{equation}\label{ksym}
  \p_1x = -y^{-1}+xy-xy^{-1}\;,\quad
  \p_1y = x^{-1}-yx+yx^{-1}\;,
\end{equation}
together with a discrete symmetry, which we will not consider here, and conjectured the integrability of
\eqref{ksym} (and of the discrete symmetry). The integrability of (\ref{ksym}) was proven in \cite{Wolf2012}, where
a Lax representation for \eqref{ksym}, as well as a symmetry for it were found; the latter is given by
\begin{equation}\label{ksym2}
  \begin{array}{l}
    \p_2x =\phantom{-}xyx+xy^2-xy-(yx)^{-1}- y^{-2}- x^2y^{-1}+xyx^{ -1}-xy^{-2}-(yxy)^{-1}-x(yxy)^{-1}\;,\\
    \p_2y =-yxy-yx^2+yx+(xy)^{-1}+x^{-2}+y^2x^{-1}-yxy^{-1}+yx^{-2}+(xyx)^{-1}+y(xyx)^{-1}\;.
  \end{array}
\end{equation}
Since $\p_i(yx-qxy)\subset\langle yx-q xy\rangle\subset\fA(q)$, for $i=1,2$, both (\ref{ksym}) and \eqref{ksym2}
define a derivation of the quantum torus $\cA_q=\fA(q)/\langle yx-q xy\rangle$.  On $\cA_q$, (\ref{ksym}) can be
written in the Heisenberg form
\begin{equation}\label{eq:p1}
  \p_1x=\frac{1}{ q-1}[\fH^{(1)},x]\;,\quad \p_1y=\frac{1}{ q-1}[\fH^{(1)},y]\;,
\end{equation}
with $ \fH^{(1)}=qx^{-1} y^{-1}+q y^{-1}+y+qx+x^{-1}$. Taking $\fH^{(2)}:=\(\fH^{(1)}\)^2-(1+q)\fH^{(1)}-4q$,
%
%
we can recast (\ref{ksym2}) also in the Heisenberg form
\begin{equation}\label{eq:p2}
  \p_2x=\frac{1}{ q^2-1}[\fH^{(2)},x]\;,\quad
  \p_2y=\frac{1}{ q^2-1}[\fH^{(2)},y]\;.
\end{equation}
This gives an alternative proof that (\ref{ksym}) and \eqref{ksym2} are derivations of $\cA_q$. Expanded, and
using the commutation relation $yx=qxy$ (which implies for example that $y^{-1}x=q^{-1}xy^{-1}$), $\fH^{(2)}$ can
be written as
\begin{align}\label{eq:H2_expand}
  \fH^{(2)}&=y^2+q^2x^2+q^2y^{-2}+x^{-2}+q^3x^{-2}y^{-2}\nonumber\\
  &\quad+(q+1)(-y-qx+qxy+qxy^{-1}+q^{-1}x^{-1}y+q^2x^{-1}y^{-2}+qx^{-2}y^{-1})\;.
\end{align}

We first consider the limits of \eqref{eq:p1} and \eqref{eq:p2} when $q\to 1$. In this case
$\cA=\bbC[x,x^{-1},y,y^{-1}]$ is the algebra of Laurent polynomials in two variables, in particular it is
commutative and we are in the case of a classical limit. Setting $q(\h)=1+\h$, the limit $\p_{\bH^{(1)}}$ is given
by
\begin{equation*}
  \p_{\bH^{(1)}}x=\pb{x^{-1} y^{-1}+ y^{-1}+y+x+x^{-1},x}=x^{-1}\pb{y^{-1},x}+\pb{y^{-1}+y,x}=xy-(x+1)y^{-1}\;,
\end{equation*}
and similarly $\p_{\bH^{(1)}}y=xy+(y+1)x^{-1}.$ We can use this result for computing $\p_{\bH^{(2)}}$ since
$\fH^{(2)}$ is a polynomial in $\fH^{(1)}$. The result is that
\begin{equation*}
  \p_{\bH^{(2)}}x=(H-1)(xy-(x+1)y^{-1})\;,\qquad  \p_{\bH^{(2)}}y=(H-1)((y+1)x^{-1}-xy)\;,
\end{equation*}
where $H$ is $\fH^{(1)}$ evaluated at $q=1$, that is $H=x^{-1} y^{-1}+ y^{-1}+y+x+x^{-1}$.

To finish, we consider the limiting derivation of $\p_2$ when $q\to-1$, setting $q=-1+\h$. We easily get from
\eqref{eq:H2_expand}
\begin{equation*}
  \frac{\fH^{(2)}}{q(\h)^2-1}=\frac1{\h}\(H^{(2)}_0+\h H^{(2)}_1\)\pmod{\cH_\h}\;,
\end{equation*}
where
\begin{align*}
  H^{(2)}_0&=-\frac1{2}\(x^2+y^2+x^{-2}+y^{-2}-x^{-2}y^{-2}\)\;,\\
  H^{(2)}_1&=-\frac12\(x-y-xy-xy^{-1}-x^{-1}y+x^{-1}y^{-2}-x^{-2}y^{-1}\)\;.
\end{align*}
Let $\bH^{(2)}:=(H^{(2)}_0,H^{(2)}_1)$. Then, in terms of the generators of $\Pi(\cA)$,
\begin{equation*}
  \bH^{(2)}=-\frac12\(X+Y+X^{-1}+Y^{-1}-X^{-1}Y^{-1}+U-V-W-Y^{-1}W+X^{-1}Y^{-1}U-X^{-1}Y^{-1}V\)\;.
\end{equation*}
Using \eqref{eq:leibniz_for} and Table \ref{tab:vol_2},
\begin{equation*}
  \p_{\bH^{(2)}}x=
  \pbm{\bH^{(2)}}x=\pp{\bH^{(2)}}Y\cdot(-2xy^2)+\pp{\bH^{(2)}}V\cdot(-2xy^2)+\pp{\bH^{(2)}}W\cdot(-2x^2y)\;.
\end{equation*}
As we already pointed out, $U$, $V$ and $W$ act trivially on $\cA$ and we get
\begin{align*}
  \p_{\bH^{(2)}}x&=(1-Y^{-2}+X^{-1}Y^{-2})\cdot xy^2-(1+X^{-1}Y^{-1})\cdot xy-(1+Y^{-1}+X^{-1})\cdot x^2y\\
    &=-y-xy+xy^2-xy^{-2}+x^{-1}y^{-2}-x^{-1}y^{-1}-x^2y-x^2y^{-1}\;,
\end{align*}
and similarly
\begin{equation*}
  \p_{\bH^{(2)}}y=x-xy-x^{-1}y^{-1}+xy^2-x^2y+x^{-1}y^2+x^{-2}y-x^{-2}y^{-1}\;.
\end{equation*}
The derivations $\p_1$ and $\p_2$ admit higher symmetries $\p_n$ of the form
\begin{equation*}
  \p_nx=\frac{1}{1-q^n}[\fH^{(n)},x]\;\quad   \p_ny=\frac{1}{1-q^n}[\fH^{(n)},y]\;,
\end{equation*}
for which the limiting derivation as $q\to\xi$, with $\xi$ an $n$-th root of unity, can be obtained in the same
way.



\end{document}